\documentclass[a4paper,11pt]{article}
\title{The higher spin Laplace operator}
\author{Hendrik De Bie\thanks{Hendrik.DeBie@UGent.be}, David Eelbode\thanks{david.eelbode@uantwerpen.be}, Matthias Roels\thanks{matthias.roels@uantwerpen.be}}
\date{}
\usepackage[a4paper]{geometry}

\usepackage{graphicx}
\usepackage{epsfig}
\usepackage{subfigure}
\usepackage{float}
\restylefloat{figure}
\restylefloat{table}
\usepackage[font=small,format=plain,labelfont=bf,up,textfont=it,up]{caption}

\usepackage{enumitem}
\usepackage[usenames,dvipsnames]{color}

\usepackage{amsmath,amsfonts,amsthm,amssymb,mathdots}
\usepackage{mathrsfs}
\usepackage{mathtools}
\usepackage{accents}
\usepackage{dsfont}
\usepackage{cases}

\usepackage{tikz}
\usetikzlibrary{arrows,matrix,positioning,shapes,fit,calc}
\usepackage[english]{babel}

\newcommand{\inner}[1]{\langle #1\rangle} 
\newcommand{\comm}[1]{\lbrack #1 \rbrack} 
\newcommand{\set}[1]{\left\{ #1 \right\} } 
\newcommand{\norm}[1]{\left| #1 \right|} 
\newcommand{\brac}[1]{\left( #1 \right)} 

\newcommand{\longhookrightarrow}{\ensuremath{\lhook\joinrel\relbar\joinrel\rightarrow}}
\newcommand{\quotient}[2]{\left.\raisebox{.2em}{$#1$}\middle/\raisebox{-.2em}{$#2$}\right.}

\newcommand{\mR}{\mathbb{R}} 
\newcommand{\mN}{\mathbb{N}} 
\newcommand{\mC}{\mathbb{C}} 
\newcommand{\mZ}{\mathbb{Z}} 

\newcommand{\D}{\partial} 
\newcommand{\mE}{\mathbb{E}} 

\newcommand{\mcA}{\mathcal{A}}

\newcommand{\mcC}{\mathcal{C}}
\newcommand{\mcD}{\mathcal{D}}
\newcommand{\mcH}{\mathcal{H}}

\newcommand{\mcJ}{\mathcal{J}}
\newcommand{\mcM}{\mathcal{M}}
\newcommand{\mcN}{\mathcal{N}}
\newcommand{\mcP}{\mathcal{P}}

\newcommand{\mcR}{\mathcal{R}}
\newcommand{\mcS}{\mathcal{S}}

\newcommand{\mcU}{\mathcal{U}}

\newcommand{\mS}{\mathbb{S}}  
\newcommand{\mV}{\mathbb{V}}

\newcommand{\f}{\mathfrak{f}}

\newcommand{\so}{\mathfrak{so}}

\newcommand{\spl}{\mathfrak{sl}}

\newcommand{\Spin}{\operatorname{Spin}}
\newcommand{\SO}{\operatorname{SO}}

\newcommand{\Aut}{\operatorname{Aut}}
\newcommand{\End}{\operatorname{End}}
\newcommand{\Span}{\operatorname{Span}}

\newcommand{\fd}{\mathfrak{f}^{\dagger}}

\newtheorem{definition}{Definition}[section]
\newtheorem{proposition}{Proposition}[section]
\newtheorem{theorem}{Theorem}[section]
\newtheorem{lemma}{Lemma}[section]
\newtheorem{remark}{Remark}[section]

\begin{document}
\maketitle

\begin{abstract}
\noindent
This paper deals with a certain class of second-order conformally invariant operators acting on functions taking values in particular (finite-dimensional) irreducible representations of the orthogonal group. These operators can be seen as a generalisation of the Laplace operator to higher spin as well as a second order analogue of the Rarita-Schwinger operator. To construct these operators, we will use the framework of Clifford analysis, a multivariate function theory in which arbitrary irreducible representations for the orthogonal group can be realised in terms of polynomials satisfying a system of differential equations. As a consequence, the functions on which this particular class of operators act are functions taking values in the space of harmonics homogeneous of degree $k$. We prove the ellipticity of these operators and use this to investigate their kernel, focusing on both polynomial solutions and the fundamental solution. 
\end{abstract}


\section{Introduction}
\label{Intro}
The aim of this paper is to introduce a framework to study a certain class of second-order conformally invariant operators which can be seen as generalisations of the classical Laplace operator. 

\noindent
Traditionally, these operators are mostly studied on either the conformal sphere $S^m$, which is a flat model of conformal geometry, or curved analogues thereof which are modelled on a principle fibre bundle \cite{Cap, Cap2, Cap3}. A useful method to study these operators is the ambient space method developed by Fefferman and Graham \cite{FG}, which requires embedding a $m$-dimensional manifold into a space of dimension $m+2$. The main idea behind their approach is that the `flat' equations on the ambient space (the ones presented in this paper) are in a suitable sense restricted to the embedded space-time, giving rise to the equations on the curved background. 

\noindent
The advantage of using the Clifford analysis model for spin fields (see below) lies in the fact that it leads to an encompassing framework in which both the dimension $m$ as the spin number can be treated as a parameter. From this point of view, the results obtained in this paper form the scalar version of the function theory for the Rarita-Schwinger operator on $\mR^m$ which was developed in \cite{BSSVL1}.  \\
\\
\noindent
The aforementioned language of Clifford analysis refers to a multivariate function theory which is often described as a generalisation of complex analysis to arbitrary {dimension} $m \in \mN$. At the very heart of this theory lies the Dirac operator $\D_x$ on $\mR^m$, a conformally invariant first order elliptic differential operator generalising both the Cauchy-Riemann operator $\partial_z$ and the operator introduced by P.A.M. Dirac in 1928 \cite{Dirac}. This operator moreover satisfies $\Delta_x=-\D_x^2$ (with $\Delta_x$ the Laplace operator on $\mR^m$), which means that Clifford analysis is a refinement of classical harmonic analysis in $m$ dimensions. We refer the reader to the standard references \cite{BDS, DSS, GM} for more information. While classical Clifford analysis is centred around the study of functions on $\mR^m$ taking values in the spinor space $\mS$, on which the Dirac operator is canonically defined, several authors have been studying generalisations of the developed techniques to the so-called higher spin theory \cite{BSSVL1, Eelbode, Eelbode2, ES, EVL}. This concerns the study of higher spin Dirac operators, acting on functions on $\mR^m$ taking values in arbitrary irreducible representations of $\Spin(m)$. So far, the theory has been focusing on first-order conformally invariant operators, reflected in the fact that the functions under consideration take their values in irreducible {\em half-integer} representations for the spin group (the spinor space being the easiest case of such a representation). This then leads to function theories refining (poly-)harmonic analysis on $\mR^m$, see \cite{ES}. As mentioned earlier, we aim at extending these results to a certain second-order conformally invariant operator acting on functions taking their values in the simplest \textit{integer} highest-weight representation. This will lead to analogues of the Rarita-Schwinger function theory (see e.g. \cite{BSSVL1}). The existence of the operators we are introducing follows from general arguments, see e.g. \cite{Br, Fegan}. However, our focus is quite different: after developing explicit expressions for the higher spin Laplace operator, we will study in depth its polynomial null-solutions as well as the fundamental solution. \\
\\ 
\noindent
The paper is organised as follows: after a brief introduction to Clifford analysis in section 2, we will construct the higher spin Laplace operator in section 3. In section 4 we will construct all polynomial solutions for this operator, whereas section 5 will be devoted to the problem of constructing the fundamental solution for our operator. Finally, in the last section we investigate the connection with the Rarita-Schwinger operator, a conformally invariant operator acting on functions taking values in the irreducible $\Spin(m)$-module with highest weight $\lambda=\brac{k+\frac{1}{2},\frac{1}{2},\ldots,\frac{1}{2},\pm\frac{1}{2}}$. We have also included an appendix to prove some results of a more technical nature from section 3 (related to the conformal invariance and ellipticity). 

\section*{Acknowledgement}
This research was supported by the Fund for Scientific Research-Flanders (FWO-V), project ``Construction of algebra realisations using Dirac-operators'', grant G.0116.13N. The results in section \ref{Construction} and \ref{PolSol} were originally obtained in the third author's (unpublished) master thesis \cite{R}. 


\section{Preliminaries on Clifford analysis} 
\label{CA}
Let us first introduce the (real) universal Clifford algebra $\mR_m$ as the algebra generated by an orthonormal basis $\{e_1,\ldots,e_m\}$ for the vector space $\mR^m$ endowed with the Euclidean inner product $\inner{u,x} = \sum_j x_ju_j$ using the multiplication rules 
\begin{equation*}
e_ae_b + e_be_a = -2\langle e_a,e_b\rangle = -2\delta_{ab} 
\end{equation*}
with $1 \leq a, b \leq m$. The complex Clifford algebra $\mC_m$ is then defined as the algebra $\mC_m = \mR_m \otimes \mC$. This algebra is $\mZ_2$-graded, and the even subalgebra (respectively the odd subspace) is denoted by means of $\mC_m^+$ (resp. $\mC^-_m$). We will not consider functions taking their values in $\mC_m$, but restrict the values to a suitable subspace, which is known as the spinor space. This space can be realised as a matrix space, as is often done by physicists, or as a subspace of the Clifford algebra (see \cite{DSS, GM}): 
\begin{definition}
In case $m = 2n$, the Witt basis $\{\f_j, \fd_j : 1 \leq j \leq n\}$ for $\mC_m$ is defined by means of
\begin{equation*}
\f_j = \frac{1}{2}\big(e_{2j-1} - ie_{2j}\big) \qquad \mbox{and} \qquad \fd_j = -\frac{1}{2}\big(e_{2j-1} + ie_{2j}\big)\ . 
\end{equation*}
The element $I := (\f_1\fd_1)(\f_2\fd_2)\ldots(\f_n\fd_n) \in \mC_m$ defines a primitive idempotent ($I^2 = I$), and in terms of this element one then introduces the complex vector spaces $\mS^\pm_{2n} := \mC^{\pm}_{2n}I$. 
\end{definition}
\noindent
The main reason why it is better to consider spinor-valued functions is the following: these spaces carry the irreducible spinor representations for the spin group or its Lie algebra $\so(m)$. They are both realised inside the Clifford algebra: 
\begin{definition}
The (real, compact) spin group $\Spin(m)$ can be defined as 
\begin{equation*}
\Spin(m) = \left\{\prod_{j = 1}^{2k}\omega_j : \omega_j \in S^{m-1}\right\} \subset \mR_m\ ,
\end{equation*}
with $S^{m-1} \subset \mR^m$ the unit sphere (viz. $\omega_j^2 = -1$). This Lie group defines a double cover for the orthogonal group, see e.g. \cite{Po}. 
\end{definition}
\begin{definition}
The orthogonal Lie algebra $\so(m)$ can be realised as the space $\mC_m^{(2)}$ of bivectors, defined as the linear hull of the elements $e_{ab} := e_ae_b$ with $1 \leq a \neq b \leq m$. This vector space becomes a Lie algebra under the commutator $[B_1,B_2] := B_1B_2 - B_2B_1$. 
\end{definition}
\noindent
In case $m = 2n$, the spaces $\mS^{\pm}_{2n}$ define inequivalent irreducible representations for the spin group and its Lie algebra (both actions given by multiplication in the Clifford algebra $\mC_{2n}$). In case $m = 2n+1$ however, there is (up to isomorphism) a unique spinor space (the so-called space of Dirac spinors) which will be denoted as $\mS_{2n+1}$ and can be explicitly realised as $\mS_{2n+1} \cong \mS^{\pm}_{2n+2}$ (whereby one can choose either the plus sign or the minus sign, because as representation spaces for $\so(2n+1)$ they are equivalent). 
\begin{remark}\label{remark2.1}
From now on we will omit the subscript attached to the spinor spaces. This subscript refers to the dimension of the underlying vector space, equal to $m \in \mN$ from now on. The superscript only matters in even dimensions (cfr. supra). In order to avoid having to drag this parity sign along in what follows, we will work with Dirac spinors in both even and odd dimensions (see next definition). 
\end{remark}
\begin{definition}
The Dirac spinor space $\mS$ stands for $\mS_{2n+1}$ in case $m = 2n+1$ is odd, and for the direct sum $\mS_{2n}^+ \oplus \mS_{2n}^-$ in case $m = 2n$ is even. Note that the latter is not irreducible as a module for the Lie algebra $\so(2n)$. 
\end{definition}
\noindent
The classical Dirac operator in $\mR^m$ is given by $\D_x=\sum_{j=1}^m e_j \partial_{x_j}$. It is the unique elliptic first-order conformally invariant differential operator acting on spinor valued functions $f(x)$ on $\mR^m$. It factorises the Laplace operator $\Delta_x=-\D_x^2$ on $\mR^m$. A spinor-valued function $f$ is monogenic in an open region $\Omega \subset \mR^m$ if and only if $\D_xf=0$ in $\Omega$. For a detailed study of the (classical) theory of monogenic functions, see \cite{BDS, DSS, GM}. In \cite{CSVL, GM} it was shown that every (finite-dimensional) irreducible representation for the spin group (or its Lie algebra) with integer (half-integer) highest weight can be realised in terms of scalar-valued harmonic (spinor-valued monogenic) polynomials of several (dummy) vector variables, a convenient alternative for the spaces of traceless tensors often used in physics. Our cases of interest are presented in the definition below, where from now on we write $\ker(\mcD_1,\ldots,\mcD_n) := \ker\mcD_1 \cap \ldots \cap \ker\mcD_n$. 
\begin{definition}
For $k\geq \ell$, the vector space of simplicial harmonic polynomials is defined as
\begin{equation*}
\mcH_{k,\ell}\brac{\mR^{2m},\mC}:=\mcP_{k,\ell}\brac{\mR^{2m},\mC} \cap \mathrm{ker}\brac{\Delta_{x},\Delta_{u},\inner{\D_u,\D_{x}},\inner{x,\D_{u}}}.
\end{equation*}
Here the space $\mcP_{k,\ell}\brac{\mR^{2m},\mC}$ is the space of $\mC$-valued polynomials depending on two vector variables $(x,u)\in \mR^{2m}$, sometimes referred to as a matrix variable. The integers $k$ and $\ell$ hereby refer to the degree of homogeneity in the variable $x$ and $u$ respectively. Recall that $\langle \cdot,\cdot \rangle$ denotes the Euclidean inner product. 
\end{definition}
\noindent
Note that this definition is not symmetric with respect to $x \leftrightarrow u$ which is reflected in the dominant weight 
condition $k \geq \ell$. Also note that for $\ell =0$, the definition reduces to the classical harmonic polynomials $\mcH_k
\brac{\mR^{m},\mC}$, see e.g. \cite{GM}. The vector space $\mcH_{k,\ell}\brac{\mR^{2m},\mC}$ is an irreducible $\Spin(m)$-representation if $m>4$. If $
\ell = 0$, this condition can even be relaxed to $m \geq 3$. The regular action of the spin group on $f\in \mcH_{k,\ell}\brac{\mR^{2m},
\mC}$ is for all $s \in \Spin(m)$ given by $H(s)\comm{f}(x,u)=f\brac{\bar{s}x s,\bar{s}u s}$, and the corresponding (derived) 
action of the orthogonal Lie algebra is given by:
\begin{align*}
\mathrm{d}H(e_{ij})\comm{f}(x,u)&=\brac{L_{ij}^x+L_{ij}^u}f(x,u) \\
:&=\brac{x_i\partial_{x_j}-x_j\partial_{x_i}+u_i\partial_{u_j}-u_j\partial_{u_i}}f(x,u).
\end{align*}
The operators $L_{ij}$ are called the angular momentum operators. Without proof (see \cite{CSVL}), we also mention that the highest weight is given by 
\begin{equation*}
\lambda=\brac{k,\ell,0,\ldots,0}=:(k,\ell),
\end{equation*}
where the length of this vector is equal to $n$ for $m \in \{2n, 2n+1\}$, and that the highest weight vector is given by 
\begin{equation*}
w_{k,\ell} := (x_1 - ix_2)^{k - \ell}\big((x_1 - ix_2)(u_3 - iu_4) - (x_3 - ix_4)(u_1 - iu_2)\big)^\ell\ .
\end{equation*}
We also mention the following definition, underlying the classical Howe dual pair SO$(m) \times \spl(2)$ for which we refer e.g. to \cite{GW, H}: 
\begin{proposition} \label{sl2}
The Laplace operator, together with its symbol, span a Lie algebra:
\begin{equation*}
\mathfrak{sl}(2)=\mathrm{Alg}\brac{X,Y,H}\cong\mathrm{Alg}\brac{-\frac{1}{2}\Delta_x,\frac{1}{2}\norm{x}^2,-\brac{\mE_x+\frac{m}{2}}}.
\end{equation*}
Here, $\mE_x=\sum_jx_j\D_{x_j} = r\D_r$ (with $r = |x|$ the norm of $x \in \mR^m$) denotes the Euler operator, measuring the degree of homogeneity in the variables $(x_1,\ldots,x_m)$. 
\end{proposition}
\noindent
In the last section, we will need the half-integer version: 
\begin{definition}
For $k\geq \ell$, the vector space of simplicial monogenic polynomials is defined as
\begin{equation*}
\mcS_{k,\ell}\brac{\mR^{2m},\mS}:=\mcP_{k,\ell}\brac{\mR^{2m},\mS} \cap \mathrm{ker}\brac{\D_{x},\D_{u},\inner{x,\D_{u}}}.
\end{equation*}
Here the space $\mcP_{k,\ell}\brac{\mR^{2m},\mS}$ is the space of spinor-valued polynomials depending on two vector variables $(x,u)\in \mR^{2m}$. The integers $k$ and $\ell$ hereby refer to the degree of homogeneity in the variable $x$ and $u$ respectively. 
\end{definition}
\noindent
This definition is again not symmetric with respect to $x \leftrightarrow u$, which also here reflects the dominant weight condition $k \geq \ell$. Also note that if $\ell =0$, the definition reduces to the classical monogenic polynomials $\mcM_k
\brac{\mR^{m},\mS}$, i.e. polynomial null solutions for $\D_x$. The vector space $\mcS_{k,\ell}\brac{\mR^{2m},\mS}$ is an irreducible $\Spin(m)$-representation if $m$ is odd and splits into two inequivalent irreducible representations if $m$ is even cfr. remark \ref{remark2.1} (see also \cite{CSVL, GM}). The action of the spin group on simplicial monogenics is given by
\begin{equation*}
L(s)\comm{f}(x,u)=sf\brac{\bar{s}xs,\bar{s}us} \quad \big(\forall s \in \Spin(m)\big)\ .
\end{equation*}
The corresponding action of the orthogonal Lie algebra is given by:
\begin{equation*}
\mathrm{d}L(e_{ij})\comm{f}(x,u)=\brac{L_{ij}^x+L_{ij}^u-\frac{1}{2}e_{ij}}f(x,u) = \left(\mathrm{d}H(e_{ij}) - \frac{1}{2}e_{ij}\right)\comm{f}(x,u).
\end{equation*}
Without proof, we also mention the highest weight 
\begin{equation*}
\lambda=\brac{k+\frac{1}{2},\ell+\frac{1}{2},\frac{1}{2},\ldots,\frac{1}{2}}=:(k,\ell)'
\end{equation*}
and the corresponding highest weight vector (see \cite{CSVL})
\begin{equation*}
v_{k,\ell} := (x_1 - ix_2)^{k - \ell}\big((x_1 - ix_2)(u_3 - iu_4) - (x_3 - ix_4)(u_1 - iu_2)\big)^\ell I = w_{k,\ell}I\ ,
\end{equation*}
where $I \in \mC_m$ is a primitive idempotent realising the spinor space as a left ideal (see e.g. \cite{DSS} for more details).
 

\section{Construction of the higher spin Laplace operator}
\label{Construction}
It was already mentioned that the Laplace operator (acting on $\mC$-valued fields) is related to the Dirac operator (acting on spinor-valued fields). The operator which generalises the Dirac operator to higher spin is the so-called Rarita-Schwinger operator (see \cite{BSSVL1,RSE} and also the last section), and in this section we will construct the generalisation of the Laplace operator to the case of higher spin:
\begin{equation*}
\mcD_k:\mcC^{\infty}\brac{\mR^m,\mcH_k}\longrightarrow \mcC^{\infty}\brac{\mR^m,\mcH_k}.
\end{equation*}
Note that the space $\mcH_k(\mR^m,\mC)$, hereby plays the role of target space. This means that elements of e.g. $\mcC^{\infty}(\mR^m,\mcH_k)$ are functions of the form $f(x,u)$, where $x \in \mR^m$ is the variable on which the operator $\mcD_k$ is meant to act, satisfying $f(x,u) \in \mcH_k(\mR^m,\mC)$ for every $x \in \mR^m$ fixed. Note that one might be tempted to think that $\mcD_k=\Delta_x$, as the relation $[\Delta_x,\Delta_u] = 0$ clearly shows that $\Delta_x$ is a rotationally invariant second-order operator preserving the values $\mcH_k$, but the operator $\Delta_x$ is not {\em conformally} invariant with respect to the action of the inversion on $\mcH_k$-valued functions (see below). This means that we will have to add extra terms to ensure conformal invariance. For example, also the second-order operator $\inner{u,\D_x}\inner{\D_u,\D_x}$ is rotationally invariant and is well-defined on $\mcH_k$-valued functions, provided we can apply a projection onto the space $\mcC^{\infty}\brac{\mR^m,\mcH_k}$ because
\begin{equation*}
\inner{u,\D_x}\inner{\D_u,\D_x}:\mcC^{\infty}\brac{\mR^m,\mcH_k}\longrightarrow \mcC^{\infty}\brac{\mR^m,\mcH_k\oplus \norm{u}^2\mcH_{k-2}},
\end{equation*}
due to the classical Fischer decomposition, see e.g. \cite{Axler} for more information. In other words, for $f(x,u)\in \mcC^{\infty}\brac{\mR^m,\mcH_k}$ we have:
\begin{equation*}
\inner{u,\D_x}\inner{\D_u,\D_x}f(x,u)=\varphi_k(x,u)+\norm{u}^2\varphi_{k-2}(x,u),
\end{equation*}
with $\varphi_i \in \mcC^{\infty}\brac{\mR^m,\mcH_i}$. The projection of $\inner{u,\D_x}\inner{\D_u,\D_x}f$ onto $\mcC^{\infty}\brac{\mR^m,\mcH_k}$ can be found as follows: 
\begin{equation*}
\Delta_u\inner{u,\D_x}\inner{\D_u,\D_u}f(x,u)=2(2k+m-4)\varphi_{k-2}.
\end{equation*}
This implies that:
\begin{equation*}
\varphi_{k-2}=\frac{1}{2(2k+m-4)}\Delta_u\inner{u,\D_x}\inner{\D_u,\D_x}f=\frac{1}{2k+m-4}\inner{\D_u,\D_x}^2f.
\end{equation*}
We finally have that 
\begin{equation*}
\varphi_k=\brac{\inner{u,\D_x}-\frac{\norm{u}^2}{2k+m-4}\inner{\D_u,\D_x}}\inner{\D_u,\D_x}f.
\end{equation*}
We could try to define $\mathcal{D}_k$ as a linear combination of this operator with the Laplace operator $\Delta_x$. Choosing the appropriate constant, this indeed works (see the appendix for more details): 
\begin{theorem}\label{HigerSpinLaplace_def}
The higher spin Laplace operator on $\mR^m$ is the unique (up to a multiplicative constant) conformally invariant second-order differential operator
\begin{equation*}
\mcD_k:\mcC^{\infty}\brac{\mR^m,\mcH_k}\longrightarrow \mcC^{\infty}\brac{\mR^m,\mcH_k},
\end{equation*}
defined by means of
\begin{equation*}
\mcD_kf(x,u)=\brac{\Delta_x-\frac{4}{2k+m-2}\brac{\inner{u,\D_x}-\frac{\norm{u}^2}{2k+m-4}\inner{\D_u,\D_x}}\inner{\D_u,\D_x}}f(x,u).
\end{equation*}
\end{theorem}
\begin{remark}
For $k=1$, we obtain the generalised Maxwell operator from \cite{ER}, i.e. 
\begin{equation*}
\mcD_1= \Delta_x - \frac{4}{m}\inner{u,\D_x}\inner{\D_u,\D_x}:\mcC^{\infty}\brac{\mR^m,\mcH_1}\longrightarrow \mcC^{\infty}\brac{\mR^m,\mcH_1}.
\end{equation*}
\end{remark}
\noindent
For the remainder of this section, we will give a sketch of the proof for the conformal invariance of $\mcD_k$, which justifies the presence of the constant $-\frac{4}{2k+m-2}$ in theorem \ref{HigerSpinLaplace_def}. For a detailed proof, we again refer the reader to the appendix. In order to explain what conformal invariance means, we need the following concept (see \cite{Eastwood}): 
\begin{definition}
An operator $\delta_1$ is a generalised symmetry for a differential operator $\mathcal{D}$ if and only if there exists another operator $\delta_2$ so that $\mathcal{D}\delta_1=\delta_2 \mathcal{D}$. Note that for $\delta_1 = \delta_2$, this reduces to the definition of a (proper) symmetry, in the sense that $[\mcD,\delta_1] = 0$. 
\end{definition}
\noindent
One then typically tries to describe the first-order generalised symmetries, as these span a Lie algebra, see e.g. \cite{Mil}. In this particular case, the first order symmetries will span a Lie algebra isomorphic to the conformal Lie algebra $\mathfrak{so}(1,m+1)$. The higher spin Laplace operator is clearly $\mathfrak{so}(m)$-invariant because it is the composition of $\mathfrak{so}(m)$-invariant operators. This means that the angular momentum operators are symmetries of the higher spin Laplace operator, in the sense that $[\mcD_k,L_{ij}^x + L_{ij}^u] = 0$. It is also easy to see that the Euler operator and the infinitesimal translations are (generalised) symmetries of $\mathcal{D}_k$, in view of the fact that $\mcD_k\mE_x = (\mE_x + 2)\mcD_k$ and $[\mcD_k,\partial_{x_j}] = 0$. Finally, there is also a special class of generalised symmetries for the operator $\mcD_k$ which can be defined in terms of the harmonic inversion for $\mathcal{H}_k$-valued functions. This is an involution mapping solutions for $\mcD_k$ to solutions for $\mcD_k$, see the crucial relation (\ref{inversion_crucial}) below. 
\begin{definition}
The harmonic inversion is a conformal transformation defined as 
\begin{align*}
\mathcal{J}_R:&\; \mcC^{\infty}\brac{\mR^m,\mcH_k}\longrightarrow \mcC^{\infty}\brac{\mR_0^m,\mcH_k} \\ &f(x,u)\mapsto\mathcal{J}_R\comm{f}(x,u):=\norm{x}^{2-m}f\brac{\frac{x}{\norm{x}^2},\frac{xux}{\norm{x}^2}}\ .
\end{align*}
Note that this inversion consists of the classical Kelvin inversion $\mcJ$ on $\mR^m$ in the variable $x$ composed with a reflection $u \mapsto \omega u \omega$ acting on the dummy variable $u$ (where $x = |x|\omega$), and satisfies $\mathcal{J}_R^2=\mathds{1}$. 
\end{definition}
\noindent
The special conformal transformations are then for all $1 \leq j \leq m$ defined as $\mathcal{J}_R\partial_{x_j}\mathcal{J}_R$, with $\mathcal{J}_R$ the harmonic inversion from above. More explicitly, \mbox{we have}: 
\begin{equation}\label{special_transfo}
\mathcal{J}_R\partial_{x_j}\mathcal{J}_R=2\inner{u,x}\D_{u_j}-2u_j\inner{x,\D_u}+\norm{x}^2\D_{x_j}-x_j\brac{2\mE_x+m-2},
\end{equation}
which follows from calculations on an arbitrary $\mcH_k$-valued function $f(x,u)$, see appendix proposition \ref{propA1}. 
This special conformal transformation and $\D_{x_j}$, the generator of translations in the $e_j$-direction, then define a model for a Lie subalgebra which is isomorphic to $\mathfrak{sl}(2)$:
\begin{proposition}
One has
\begin{equation*}
\mathfrak{sl}(2)\cong\mathrm{Alg}\brac{\mathcal{J}_R\partial_{x_j}\mathcal{J}_R,\partial_{x_j},2\mE_x+m-2}\ .
\end{equation*}
\end{proposition}
\noindent
\begin{proof}
The classical commutator relations are verified after some straightforward computations. 
\end{proof}
\begin{proposition}
The special conformal transformations $\mathcal{J}_R\partial_{x_j}\mathcal{J}_R$ from equation (\ref{special_transfo}), with \mbox{$1\leq j \leq m$}, are generalised symmetries of the higher spin Laplace operator.
 \label{special conformal}
\end{proposition}
\noindent
\begin{proof}
This result can be proved by calculating the commutator $\comm{\mathcal{D}_k,\mathcal{J}_R\D_{x_j}\mathcal{J}_R}$, for which we refer to proposition \ref{special conformal2} in the appendix. In particular, it was shown that  
\begin{align}\label{inversion_crucial}
\mcJ_R \mcD_k \mcJ_R & = \norm{x}^4 \mcD_k\ , 
\end{align}
which generalises the classical result $\mcJ \Delta_x \mcJ = \norm{x}^4 \Delta_x$ for the Laplace operator $\Delta_x$ acting on $\mC$-valued functions (with $\mcJ$ the classical Kelvin inversion), see e.g. \cite{Axler}. 
\end{proof}
\noindent
The conformal invariance can be summarized in the following theorem:
\begin{theorem}
The first-order generalised symmetries of the higher spin Laplace operator $\mathcal{D}_k$ are given by: 
\begin{enumerate}[label=(\roman{*})]
\item The infinitesimal rotations $L_{ij}^x + L_{ij}^u$, with $1 \leq i < j \leq m$. 
\item The shifted Euler operator $(2\mE_x+m-2)$. 
\item The infinitesimal translations $\D_{x_j}$, with $1 \leq j \leq m$. 
\item The special conformal transformations $\mathcal{J}_R\D_{x_j}\mathcal{J}_R$, with $1 \leq j \leq m$. 
\end{enumerate}
These operators span a Lie algebra which is isomorphic to the conformal Lie algebra $\mathfrak{so}(1,m+1)$, whereby the Lie bracket is the ordinary commutator. 
\end{theorem}
\noindent
Besides conformal invariance, we also have another crucial property of the higher spin Laplace operator. For more details, we refer to theorem \ref{Elliptic} in the appendix. We here only mention the main conclusion: 
\begin{theorem}\label{surjective}
The higher spin Laplace operator $\mathcal{D}_k$ is an elliptic operator if $m>4$, which implies surjectivity of the map
\begin{equation*}
\mathcal{D}_k:\mcC^{\infty}\brac{\mR^m,\mcH_k}\longrightarrow \mcC^{\infty}\brac{\mR^m,\mcH_k}.
\end{equation*}  
\end{theorem}
\noindent
To conclude this section, we will introduce two other conformally invariant operators which are called (dual) twistor operators. These operators are used in the next section to describe the structure of the space of polynomial null solutions of the higher spin Laplace operator. The twistor operators are defined as follows: 
\begin{definition}
The twistor operator is the unique (up to a multiplicative constant) conformally invariant operator defined as
\begin{equation*}
\pi_k \inner{u,\D_x}:\mcC^{\infty}\brac{\mR^m,\mcH_{k-1}}\longrightarrow \mcC^{\infty}\brac{\mR^m,\mcH_k},
\end{equation*}
where $\pi_k$ is a projection on $\ker\Delta_u$. The dual twistor operator is the unique (up to a multiplicative constant) conformally invariant operator given by
\begin{equation*}
\inner{\D_u,\D_x}:\mcC^{\infty}\brac{\mR^m,\mcH_k}\longrightarrow \mcC^{\infty}\brac{\mR^m,\mcH_{k-1}}.
\end{equation*}
\end{definition}
\noindent
Since $\inner{\D_u,\D_x}f\in \mcC^{\infty}\brac{\mR^m,\mcH_{k-1}} $ for a function $f\in \mcC^{\infty}\brac{\mR^m,\mcH_k}$, the calculation at the beginning of this section shows that the twistor operator is explicitly given by 
\begin{equation*}
\pi_k\inner{u,\D_x}:=\inner{u,\D_x}-\frac{\norm{u}^2}{2k+m-4}\inner{\D_u,\D_x}.
\end{equation*}
This means that the higher spin Laplace operator is in fact a combination of the Laplace operator and the composition of the twistor operator with the dual twistor operator, i.e. 
\begin{equation*}
\mcD_k=\Delta_x-\frac{4}{2k+m-2}\pi_k\inner{u,\D_x}\inner{\D_u,\D_x}.
\end{equation*}


\section{Polynomial null solutions} 
\label{PolSol}   

In this section we study $\ell$-homogeneous polynomial solutions for $\mcD_k$, i.e. polynomials $f(x,u)$ in two vector variables satisfying $\mcD_kf(x,u)=0$ and $\mE_x f(x,u)=\ell f(x,u)$. The vector space of null solutions will from now on be denoted by $\ker_\ell\mcD_k$ and this space is not irreducible as a module for $\so(m)$, in contrast to the kernel of the Laplace operator. Note that $\bigoplus_{\ell}\ker_\ell\mcD_k$ {\em does} define an irreducible module for the real form $\so(1,m+1)$, but this is beyond the scope of the present paper. 

\noindent
First, we use theorem \ref{surjective}, to compute the dimension of $\ker_\ell\mcD_k$ as follows ($\ell \geq k$): 
\begin{equation*}
\mathrm{dim}\brac{\ker_\ell\mcD_k}=\mathrm{dim}\brac{\mcP_{\ell}\brac{\mR^m,\mcH_k}}-\mathrm{dim}\brac{\mcP_{\ell-2}\brac{\mR^m,\mcH_k}}\ .
\end{equation*}
It thus follows for all $\ell \in \mN \backslash \set{0,1}$ that 
\begin{align*}
\mathrm{dim}\brac{\ker_\ell\mcD_k}&=\left(\left(\!
    \begin{array}{c}
      m+\ell-1 \\
      m-1
    \end{array}
  \!\right)-\left(\!
    \begin{array}{c}
      m+\ell-3 \\
      m-1
    \end{array}
  \!\right)\right)\mathrm{dim}\brac{\mcH_k}\ .
\end{align*}
As this number coincides with $\mathrm{dim}\brac{\mcH_\ell\otimes \mcH_k}$, this suggests that the space $\ker_\ell\mcD_k$ can be decomposed as follows for arbitrary $\ell \geq k$ (hereby referring to \cite{Klimyk} for the tensor product decomposition rules): 
\begin{equation}\label{Decomp}
\ker_\ell\mcD_k\cong \bigoplus_{i=0}^{k} \bigoplus_{j=0}^{k-i} \mcH_{\ell-i+j,k-i-j}.
\end{equation}
A straightforward subspace of $\ker_\ell\mcD_k$, is the space of harmonics in $x$ and $u$ intersected with the kernel of the dual twistor operator $\inner{\D_u,\D_x}$. This space is known as the space of {\em Howe harmonics}, see \cite{GW, H}, and it allows for the following decomposition into $\so(m)$-irreducible summands (where we use our explicit model involving simplicial harmonics):
\begin{equation*}
\mcA_{\ell,k}\brac{\mR^{2m},\mC}:=\mcP_{\ell,k}\brac{\mR^{2m},\mC}\cap \ker\brac{\Delta_u,\Delta_x,\inner{\D_u,\D_x}}=\bigoplus_{j=0}^{k}\inner{u,\D_x}^j\mcH_{\ell+j,k-j}\brac{\mR^{2m},\mC}
\end{equation*}
It has the same structure as the space of double monogenics, and can thus be seen as the analogue of the type A solutions for the Rarita-Schwinger operator discussed in \cite{BSSVL1}. As these solutions are also killed by the operator $\inner{\D_u,\D_x}$, which then implies that the operator $\mcD_k$ reduces to the ordinary Laplace operator $\Delta_x$, they can also be described as the solutions for $\mcD_k$ in the Lorentz gauge, in analogy with what is usually done in physics. 

\noindent
We can prove by induction that decomposition (\ref{Decomp}) is in fact the correct one. To do so, we need the following lemma:
\begin{lemma}
The following operator identity holds:
\begin{equation*}
\inner{\D_u,\D_x}\mcD_k=\frac{2k+m-6}{2k+m-2}\mcD_{k-1}\inner{\D_u,\D_x}.
\end{equation*}
In particular, this means that the dual twistor operator $\inner{\D_u,\D_x}$ maps solutions of $\mcD_k$ to solutions of $\mcD_{k-1}$.
\end{lemma}
\begin{proof}
This can be proved by direct computations.
\end{proof}
\noindent
As a consequence of this lemma, the irreducible components of $\ker_{\ell-1}\mcD_{k-1}$ also appear in decomposition (\ref{Decomp}). In particular, we obtain the following result:
\begin{lemma} 
The vector space $\ker_\ell\mcD_k$ has the following decomposition:
\begin{equation*}
\ker_\ell\mcD_k\cong\mcA_{\ell,k}\oplus \ker_{\ell-1}\mcD_{k-1}.
\end{equation*}
\end{lemma}
\begin{proof}
Since $\mcA_{\ell,k}\subset \ker\inner{\D_u,\D_x}$, we have that $\mcA_{\ell,k}\cap \ker_{\ell-1}\mcD_{k-1}=\set{0}$. The result then follows from the previous lemma and some basic linear algebra.
\end{proof}
\noindent
The irreducible components of the subspace $\mcA_{\ell-1,k-1}\brac{\mR^{2m},\mC}\subset \ker_{\ell-1}\mcD_{k-1}$ of type A solutions of $\mcD_{k-1}$, i.e.
\begin{equation*}
\mcA_{\ell-1,k-1}\brac{\mR^{2m},\mC}\cong\bigoplus_{j=0}^{k-1}\mcH_{\ell-1+j,k-1-j}\brac{\mR^{2m},\mC},
\end{equation*} 
must also appear in decomposition (\ref{Decomp}). The components of the subspace $\mcA_{\ell-2,k-2}$ of $\ker_{\ell-2}\mcD_{k-2}$ must also appear in decomposition of $\ker_{\ell-1}\mcD_{k-1}$ into irreducible components and therefore, they must also occur in the decomposition of $\ker_{\ell}\mcD_{k}$. This reasoning can be repeated until we find a single component: $\mcH_{\ell-k}=\ker_{\ell-k}\mcD_0$. Putting al the irreducible summands together, we obtain decomposition (\ref{Decomp}). 
\\ \\ \noindent 
Now that we found the decomposition of $\ker_{\ell}\mcD_k$ into irreducible representations, we still need to find the operators that embed each of these irreducible summands into $\ker_{\ell}\mcD_k$. As the Laplace operator $\Delta_x$ commutes with $\mcD_k$, the application 
of $\Delta_x$ preserves solutions of $\mcD_k$ and lowers the degree of homogeneity in $x$ by two. Since the higher spin Laplace operator is conformally 
invariant, the inversion $\mcJ_R$ preserves solutions and so does $\mcJ_R\Delta_x\mcJ_R$, raising the degree of homogeneity in $x$ by two. This operator is 
explicitly given by:
\begin{align*}
\mcJ_R\Delta_x\mcJ_R=&\norm{x}^4\Delta_x+4(2k+m-4)\pi_k\inner{u,x}\inner{x,\D_u} \\
&+4\norm{x}^2\brac{\pi_k\inner{u,x}\inner{\D_x,\D_u}-\pi_k\inner{u,\D_x}\inner{x,\D_u}},
\end{align*}
where $\pi_k\inner{u,x}$ is the (principle) symbol of the twistor operator $\pi_k\inner{u,\D_x}$. Using this operator and the twistor operator $\pi_k\inner{u,\D_x}$, which acts as $\inner{u,\D_x}$ on simplicial harmonics, we have the following conclusion (recall that it only holds for $m>4$ as we heavily relied on theorem \ref{surjective} to obtain this result):
\begin{theorem}\label{ker_decomp}
The vector space $\ker_\ell\mcD_k$ has the following decomposition into irreducible modules (for the action of the orthogonal group):
\begin{equation*}
\ker_\ell\mcD_k\ =\ \bigoplus_{i=0}^{k} \bigoplus_{j=0}^{k-i} (\mcJ_R\Delta_x\mcJ_R)^i\inner{u,\D_x}^{i+j}\mcH_{\ell-i+j,k-i-j}.
\end{equation*}
Also note that $\ker_0 \mcD_k = \mcH_k(\mR^m,\mC)$.
\end{theorem}
\noindent
This result can also be visualised as the triangle in Figure \ref{firstkernel}. 
\begin{figure}[H]
\centering
\begin{tikzpicture}[>=angle 90]
\coordinate[label=above:$\phantom{a}$] (D) at (0,-1.5);
\coordinate[label=below left:$\phantom{a}$] (E) at (3.5,1.2);
\coordinate[label=below right:$\phantom{a}$] (F) at (-3.5,1.2);
\coordinate[label=below left:$\phantom{a}$] (G) at (-2.5,1.2);
\coordinate[label=below right:$\phantom{a}$] (H) at (.5,-1.1);
\coordinate[label=below left:$\phantom{a}$] (A) at (-1.5,0);
\coordinate[label=below left:$\mcA_{\ell,k}{:=}\mcP_{\ell,k}\cap \ker\begin{pmatrix} \Delta_u \\ \Delta_x \\ \inner{\D_u,\D_x}  \end{pmatrix}$] (B) at (-2,-0.25);
\coordinate[label=below right: $\mcH_{\ell,k}$] (C) at (0.8,-1.5);
\coordinate[label=above: $\ker_{\ell-1}\mcD_{k-1}$] (Z) at (0.3,0);
\fill[black] (0,-1.1) circle (2pt);
\draw (D) -- (E) -- (F)-- cycle;
\draw (G) -- (H);
\path[>=stealth,->] (A) edge [bend right=20] (B);
\path[>=stealth,->] (0,-1.1) edge [bend left=20] (C);
\end{tikzpicture}
\caption{The decomposition of $\ker_\ell\mcD_k$ under $\so(m)$}
\label{firstkernel}
\end{figure}
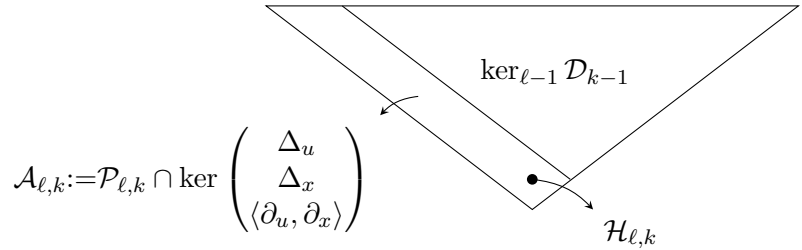
\noindent
\noindent
\begin{remark}
The case where $\ell < k$ is the degenerate case. This means that some of the summands, which do not satisfy the dominant weight condition, will be missing.
\end{remark}
\noindent
As we will see in the following theorem, each null solution of the higher spin Laplace operator is a solution of $\Delta_x^{k+1}$. As a matter of fact, the higher spin Laplace operator factorises this specific power of the Laplace operator. 
\begin{theorem}
There exist a differential operator $A_{2k}$ of order $2k$ with $k > 0$ acting between spaces of $\mcH_k$-valued functions such that:
\begin{equation*}
A_{2k}\mcD_k=\Delta_x^{k+1}=\mcD_k A_{2k}.
\end{equation*}
\end{theorem}
\begin{proof}
We will prove this inductively. For $k=1$, we can write 
\begin{equation*}
\Delta_x^2=\brac{\Delta_x+\frac{4}{m-4}\inner{u,\D_x}\inner{\D_u,\D_x}}\mcD_1
\end{equation*}
Suppose we have constructed an operator $A_{2k-2}$ such that $A_{2k-2}\mcD_{k-1}=\Delta_x^{k}=\mcD_{k-1} A_{2k-2}$.
We then have:
\begin{align*}
\Delta_x^{k+1}&=\Delta_x^k\brac{\mcD_k+\frac{4}{2k+m-2}\pi_k\inner{u,\D_x} \inner{\D_u,\D_x}} \\
&=\Delta_x^k \mcD_k +\frac{4}{2k+m-2}\pi_k\inner{u,\D_x}\Delta_x^k\inner{\D_u,\D_x} \\
&=\Delta_x^k \mcD_k +\frac{4}{2k+m-2}\pi_k\inner{u,\D_x}A_{2k-2}\mcD_{k-1}\inner{\D_u,\D_x}\\
&=\brac{\Delta_x^k+\frac{4}{2k+m-6}\pi_k\inner{u,\D_x}A_{2k-2}\inner{\D_u,\D_x}}\mcD_k,
\end{align*}
from which the operator $A_{2k}$ follows by induction. 
\end{proof}


\section{Fundamental solution} 
\label{FunSol}
Before turning to the fundamental solution of $\mcD_k$, we will first consider the fundamental solution of the Laplace operator. It is given by
\begin{equation*}
N(x) = 
  \begin{dcases*}
    \frac{1}{(2-m)A_m}\norm{x}^{2-m} & \text{if $m>2$}\\
    \frac{1}{2\pi}\log{\norm{x}} &  \text{if $m=2$,}
  \end{dcases*}
\end{equation*}
where $A_m$ is the surface area of the unit sphere $S^{m-1}$. The fundamental solution of the Laplace operator is a harmonic function in $\mcC^{\infty}\brac{\mR_0^m,\mC}$, i.e. $\Delta_x N(x)=0$ for all $x\neq 0$. 
As $\Delta_x=-\D_x^2$, the fundamental solution $E(x)$ for the Dirac operator is easily obtained through the following \cite{BDBDS, BDS}: 
\begin{equation*}
E(x)=-\D_xN(x)=- \frac{1}{A_m}\frac{x}{\norm{x}^{m}}.
\end{equation*}
This expression is the so-called Cauchy kernel and, as fundamental solution for the Dirac operator, it satisfies the relation $\D_xE(x)=\delta(x)$. Note that $E(x)\in \mcC^{\infty}\brac{\mR^m_0,\mC_m}$ and because the Clifford algebra $\mC_m$ can be seen as the space of endomorphisms of the spinor space $\mS$, we thus have that $E(x)\in \mcC^{\infty}\brac{\mR^m_0,\End(\mS)}$. 
In case of the higher spin Laplace operator, which acts on functions taking values in $\mathcal{H}_k$, the fundamental solution is then expected to belong to the function space $\mathcal{C}^{\infty}\brac{\mR_0^m,\End\brac{\mcH_k}}$. We also refer to \cite{Eelbode}, where this was done for the case of invariant operators acting on half-integer higher spin fields. We start our search for the fundamental solution of $\mcD_k$ with the following observation:
\begin{proposition}\label{Prop4}
For every $H_k(u) \in \mcH_k\brac{\mR^m,\mC}$, the function 
\begin{equation}\label{label1}
E_k(x;u):= \norm{x}^{2-m}H_k(\omega u \omega) = \norm{x}^{2-m-2k}H_k(xux),
\end{equation}
where $\omega=\norm{x}$, belongs to $\mcC^{\infty}\brac{\mR_0^m,\mathcal{H}_k}$ and has a singularity of degree $(2-m)$ in the origin $x=0$ in $\mR^m$. Furthermore, $E_k(x;u)$ belongs to the kernel of the operator $\mathcal{D}_k$ (for $x \neq 0$).
\end{proposition}
\begin{proof}
It is clear that $E_k(x;u)\in\mcC^{\infty}\brac{\mR_0^m,\mathcal{H}_k}$, as it is homogeneous of degree $k$ in the dummy variable $u \in \mR^m$. In order to prove that (\ref{label1}) belongs to the kernel of the higher spin Laplace operator for $x\neq 0$ we will rely on the fact that $\mathcal{H}_k$ is an irreducible $\Spin(m)$-representation generated by the highest weight vector $\inner{u,2\f_1}^k$. As $\mcD_k$ is a $\Spin(m)$-invariant operator, it will be sufficient to prove the statement for 
\begin{equation*}
\Phi(x;u):= \norm{x}^{2-m-2k}\inner{xux,2\f_1}^k = \norm{x}^{2-m-2k}\brac{\norm{x}^2\inner{u,2\f_1}-2\inner{u,x}\inner{x,2\f_1}}^k,
\end{equation*}
 where we used the fact that $xux=u\norm{x}^2-2\inner{u,x}x$. To calculate the action of $\mcD_k$ on $\Phi(x;u)$, we use the following relations, which can be verified by direct calculations:
 \begin{align*}
 \Delta_x\norm{x}^{\alpha}=&\: \alpha\brac{m+\alpha-2}\norm{x}^{\alpha-2} \\
 \Delta_x\inner{xux,2\f_1}^k=&\: 4k(k-1)\norm{x}^{2-m-2k}\norm{u}^2\inner{x,2\f_1}^2\inner{xux,2\f_1}^{k-2} \\
 &+2k(2k+m-4)\norm{x}^{2-m-2k}\inner{xux,2\f_1}^{k-1}\inner{u,2\f_1} \\
 \inner{\D_u,\D_x}\inner{xux,2\f_1}^k=&-2k(2k+m-2)\inner{x,2\f_1}\inner{xux,2\f_1}^{k-1}.
 \end{align*}
The action of the Laplace operator on $\Phi(x,u)$ for $x\neq 0$ then gives:
\begin{align*}
\Delta_x\Phi(x;u)=&-2k(2-m-2k)\norm{x}^{-m-2k}\inner{xux,2\f_1}^k \\
&+2k(2k+m-4)\norm{x}^{2-m-2k}\inner{xux,2\f_1}^{k-1}\inner{u,2\f_1} \\
&+4k(k-1)\norm{x}^{2-m-2k}\norm{u}^2\inner{x,2\f_1}^2\inner{xux,2\f_1}^{k-2} \\
&+2\sum_j\brac{\D_{x_j} \norm{x}^{2-m-2k}}\brac{\D_{x_j}\inner{xux,2\f_1}^k} \\
=&\: 2k(2-m-2k)\norm{x}^{-m-2k}\inner{xux,2\f_1}^k \\
&+2k(2k+m-4)\norm{x}^{2-m-2k}\inner{xux,2\f_1}^{k-1}\inner{u,2\f_1} \\
&+4k(k-1)\norm{x}^{2-m-2k}\norm{u}^2\inner{x,2\f_1}^2\inner{xux,2\f_1}^{k-2}.
\end{align*}
The action of $\inner{u,\D_x}\inner{\D_u,\D_x}$ is given by:
\begin{align*}
\inner{u,\D_x}\inner{\D_u,\D_x}\Phi(x;u)=&-k(2k+m-2)\inner{u,\D_x}\inner{x,2\f_1}\norm{x}^{2-m-2k}\inner{xux,2\f_1}^{k-1} \\
=&\: k(2k+m-2)^2\inner{x,2\f_1}\inner{u,x}\norm{x}^{-m-2k}\inner{xux,2\f_1}^{k-1} \\
&-k(2k+m-2)\inner{u,2\f_1}\norm{x}^{2-m-2k}\inner{xux,2\f_1}^{k-1} \\
&+2k(k-1)(2k+m-2)\norm{u}^2\inner{x,2\f_1}^2\norm{x}^{2-m-2k}\inner{xux,2\f_1}^{k-2},
\end{align*}
whereas the action of the last term from $\mcD_k$ leads to
\begin{equation*}
\norm{u}^2\inner{\D_u,\D_x}\Phi(x,u)=k(k-1)(2k+m-2)(2k+m-4)\norm{u}^2\inner{x,2\f_1}^2\norm{x}^{2-m-2k}\inner{xux,2\f_1}^{k-2}.
\end{equation*}
Putting everything together with the proper constants gives zero, which proves the statement. 
\end{proof}
\noindent
Note that up until now, we have excluded the point-wise singularity of $E_k(x;u)$ at $x=0$. In order to investigate this singularity, we will use results from distribution theory (more precisely, we will use the Riesz potentials on $\mR^m$). Take arbitrary $\alpha \in \mC$ fixed and consider the function $E_k^{\alpha}(x;u):=\norm{x}^{\alpha-2k}\inner{xux,2\f_1}^k$. This is then again a function in $\mcC^{\infty}\brac{\mR_0^m,\mathcal{H}_k}$. Under the action of the higher spin Laplace operator, tedious but similar calculations as the ones in the proof of proposition \ref{Prop4} give: 
\begin{equation}\label{label2}
\begin{split}
\mcD_kE_k^{\alpha}(x;u)=&(\alpha+m-2)\brac{\alpha+\frac{4k}{2k+m-2}}\norm{x}^{\alpha-2k-2}\inner{xux,2\f_1}^k \\
+&(\alpha+m-2)(\alpha+m)\frac{4k}{2k+m-2}\inner{u,x}\inner{x,2\f_1}\norm{x}^{\alpha-2k-2}\inner{xux,2\f_1}^{k-1} \\
+&\frac{4k(k-1)(\alpha+m)(\alpha+m-2)}{(2k+m-2)(2k+m-4)} \norm{u}^2\inner{x,2\f_1}^2\norm{x}^{\alpha-2k}\inner{xux,2\f_1}^{k-2}
\end{split}
\end{equation}
For $\alpha=2-m$, we again observe that $E_k^{\alpha}(x;u)$ is in the kernel of $\mathcal{D}_k$ and has a pointwise singularity of degree $(2-m)$ at the origin. As $E_k^{\alpha}(x;u)$ belongs to the space $L_1^{\mathrm{loc}}\brac{\mR^m,\mathcal{H}_k}$ of $\mcH_k$-valued locally integrable functions on $\mR^m$ for $\alpha \in \mC$ with $\mathfrak{R}(\alpha) > -m-2$, it thus defines a distribution on the space $\mathcal{D}\brac{\mR^m,\mathcal{H}_k}$ of $\mathcal{H}_k$-valued test functions (functions in $\mcC^{\infty}\brac{\mR^m,\mathcal{H}_k}$ with compact support). For arbitrary for $\gamma$ with $\mathfrak{R}(\gamma)>-m$, we then consider the distribution $\norm{x}^{\gamma}$ whose action on test functions $\phi \in\mathcal{D}\brac{\mR^m}$ is defined by
\begin{equation*}
\inner{\norm{x}^{\gamma},\phi}=\int_{\mR^m}\norm{x}^{\gamma}\phi(x)\mathrm{d}x\ .
\end{equation*}  
The following result will be used, see e.g. \cite{Helgason}:
\begin{lemma}\label{Helg}
The mapping $\gamma \mapsto \norm{x}^{\gamma}$ can be uniquely extended to a meromorphic mapping from $\mC$ to the space of tempered distributions on $\mR^m$ (i.e. holomorphic on $\mC$, except for a few isolated points). The poles are the points $\gamma=-m-2a$ (for all $a\in \mN$) and they are all simple. 
\end{lemma}
\noindent
For $\gamma \in \mC\setminus\set{m+2a,-2b: a,b\in \mN}$ we then introduce the Riesz potential by means \mbox{of (see \cite{Helgason} page 135-136 for more details)}
\begin{equation*}
I_x^{\gamma}:=\frac{\Gamma\brac{\frac{m-\gamma}{2}}}{2^{\gamma}\pi^{\frac{m}{2}}\Gamma\brac{\frac{\gamma}{2}}}\norm{x}^{-m+\gamma}.
\end{equation*}
This distribution acts on test functions $\phi$ through a convolution product, and one has that $I_x^{0}\phi=\lim_{\gamma \to 0}I_x^{\gamma}\phi=\phi(0)$. Note that the poles of $\norm{x}^{-m+\gamma}$ are cancelled by the poles of $\Gamma\brac{\frac{\gamma}{2}}$. The Riesz potential for $\gamma=2$ can be seen as some sort of inverse of the Laplace operator $\Delta_x$, because it satisfies $I_x^{\gamma}\Delta_x\phi=\Delta_xI_x^{\gamma}\phi=-I_x^{\gamma-2}\phi$ in distributional sense. A repeated application of this relation then leads to the relation $I_x^{\gamma}=(-1)^a\Delta_x^aI_x^{\gamma+2a}$, for all $a \in \mN$, so we can define 
\begin{equation*}
I_x^{-2a}=(-1)^a\Delta_x^a\delta\brac{x},
\end{equation*}
with $\delta\brac{x}$ the Dirac-delta distribution on $\mR^m$. This thus defines an analytic continuation of the mapping $\gamma \mapsto I_x^{\gamma}$ to a holomorphic function with poles in $\set{\gamma=m+2a: a\in \mN}$, which are then precisely the poles of $\Gamma\brac{\frac{m-\gamma}{2}}$. Our findings can be reformulated, in the sense that we can analytically extend the mapping $\gamma \mapsto \norm{x}^{-m+\gamma}$ to $\mC\setminus\set{-2a: a\in \mN}$, according to lemma \ref{Helg}. Its singularities are the simple poles, with residues
\begin{equation*}
\underset{\gamma=-2a}{\mathrm{Res}}\norm{x}^{-m+\gamma}=\underset{\gamma=-2a}{\mathrm{Res}}\brac{\frac{2^{\gamma}\pi^{\frac{m}{2}}\Gamma\brac{\frac{\gamma}{2}}}{\Gamma\brac{\frac{m-\gamma}{2}}}I_x^{\gamma}} 
=\frac{2^{-2a}\pi^{\frac{m}{2}}}{\Gamma\brac{\frac{m}{2}+a}}\underset{\gamma=-2a}{\mathrm{Res}}\brac{\Gamma\brac{\frac{\gamma}{2}}}I_x^{-2a}.
\end{equation*}
In view of the fact that 
\begin{equation*}
\underset{\gamma=-2a}{\mathrm{Res}}\Gamma\brac{\frac{\gamma}{2}}=\lim_{\gamma \to -2a}(\gamma +2a)\Gamma\brac{\frac{\gamma}{2}}=2\frac{(-1)^a}{a!},
\end{equation*}
we then find
\begin{equation*}
\underset{\gamma=-2a}{\mathrm{Res}}\norm{x}^{-m+\gamma}=\frac{2^{-2a+1}\pi^{\frac{m}{2}}}{\Gamma\brac{\frac{m}{2}+a}a!}\Delta_x^a\delta\brac{x}.
\end{equation*}
This implies that the mapping $\alpha \mapsto E_1^{\alpha}$ is holomorphic in $\mC \setminus \set{-m-2a: a\in \mN}$. Moreover, the poles at the values $\set{-m-2,\ldots,-m+2,-m}$ are removable singularities so the following proposition was proved:
\begin{proposition}
The mapping $\alpha \mapsto E_1^{\alpha}(x;u)=\norm{x}^{\alpha-2}H_1(x\:u\:x)$ can be holomorphically extended to the set $\mC \backslash \set{-m-2a: a\in \mN}$.
\end{proposition}
\noindent
This means that equation (\ref{label2}) holds in distributional sense, for $\mathfrak{R}(\alpha) > -m-1$. Hence, with this restriction on $\alpha$, we find that
\begin{align*}
\mcD_k\Phi(x,u)=&\lim_{\alpha \to 2-m}(\alpha+m-2)\brac{\alpha+\frac{4k}{2k+m-2}}\norm{x}^{\alpha-2k-2}\inner{xux,2\f_1}^k  \\
&+ \lim_{\alpha \to 2-m} (\alpha+m-2)\frac{4k(\alpha+m)}{2k+m-2}\norm{x}^{\alpha-2k-2}\inner{u,x}\inner{x,2\f_1} \inner{xux,2\f_1}^{k-1}  \\
&+ \lim_{\alpha \to 2-m} \frac{4k(k-1)(\alpha+m)(\alpha+m-2)}{(2k+m-2)(2k+m-4)} \norm{x}^{\alpha-2k}\norm{u}^2\inner{x,2\f_1}^2\inner{xux,2\f_1}^{k-2},  
\end{align*}
which leads to
\begin{equation}\label{label3}
\begin{split}
\mcD_k\Phi(x,u)=&\brac{2-m+\frac{4k}{2k+m-2}}\frac{2^{-2k+1}\pi^{\frac{m}{2}}}{\Gamma\brac{\frac{m}{2}+k}k!}\Delta_x^k\delta\brac{x} \inner{xux,2\f_1}^k  \\
+&\frac{8k}{2k+m-2}\frac{2^{-2k+1}\pi^{\frac{m}{2}}}{\Gamma\brac{\frac{m}{2}+k}k!}\Delta_x^k\delta\brac{x}\inner{u,x}\inner{x,2\f_1} \inner{xux,2\f_1}^{k-1}  \\ 
+&\frac{8k(k-1)}{(2k+m-2)(2k+m-4)}\frac{2^{-2k+3}\pi^{\frac{m}{2}}}{\Gamma\brac{\frac{m}{2}+k-1}(k-1)!}\Delta_x^{k-1}\delta\brac{x} \\
&\times \norm{u}^2\inner{x,2\f_1}^2\inner{xux,2\f_1}^{k-2}.
\end{split}
\end{equation} 
In view of the fact that $\inner{\delta,\phi}=\phi(0)$, we get:
\begin{align*}
\inner{(\Delta_x^k\delta\brac{x})\inner{xux,2\f_1}^k,\phi}&=\inner{(\Delta_x^k\delta\brac{x}),\inner{xux,2\f_1}^k\phi} \\
&=\inner{\delta\brac{x},\Delta_x^k\inner{xux,2\f_1}^k\phi} \\
&=\inner{\delta\brac{x},\brac{\Delta_x^k\inner{xux,2\f_1}^k}\phi+\ldots},
\end{align*}
where the dots indicate all the other terms coming from the action of the operator $\Delta_x^k$. They can be safely ignored, since we still have to act with the distribution $\delta\brac{x}$ which will make all of these terms disappear. We thus get that
\begin{equation*}
\inner{(\Delta_x^k\delta\brac{x})\inner{xux,2\f_1}^k,\phi}=\inner{\delta\brac{x},\brac{\Delta_x^k\inner{xux,2\f_1}^k}\phi}= \inner{[\Delta_x\inner{xux,2\f_1}^k]\delta\brac{x},\phi}.
\end{equation*}
A similar reasoning can be applied for the other two terms which means that equation (\ref{label3}) reduces to
\begin{equation}\label{label4}
\begin{split}
\mcD_k\Phi(x,u)=&\brac{2-m+\frac{4k}{2k+m-2}}\frac{2^{-2k+1}\pi^{\frac{m}{2}}}{\Gamma\brac{\frac{m}{2}+k}k!}\brac{\Delta_x^k\inner{xux,2\f_1}^k}\delta\brac{x}   \\
+&\frac{8k}{2k+m-2}\frac{2^{-2k+1}\pi^{\frac{m}{2}}}{\Gamma\brac{\frac{m}{2}+k}k!}\brac{\Delta_x^k\inner{u,x}\inner{x,2\f_1} \inner{xux,2\f_1}^{k-1}}\delta\brac{x}  \\ 
+&\frac{8k(k-1)}{(2k+m-2)(2k+m-4)}\frac{2^{-2k+3}\pi^{\frac{m}{2}}}{\Gamma\brac{\frac{m}{2}+k-1}(k-1)!} \norm{u}^2\\
&\times \brac{\Delta_x^{k-1}\inner{x,2\f_1}^2\inner{xux,2\f_1}^{k-2}}\delta\brac{x}.
\end{split}
\end{equation} 
In order to calculate the action of $\Delta_x^k$ on the given polynomials, we need a few lemmas:
\begin{lemma}
For $H_a(x)\in \mcH_a$ and $H_b(x)\in \mcH_b$, where $a\geq b$, one has that $H_a(x)H_b(x) \in \ker \Delta_x^{b+1}$, i.e. 
\begin{equation*}
H_a(x)H_b(x)\in \mcH_{a+b}\oplus \norm{x}^2\mcH_{a+b-2}\oplus \ldots \oplus \norm{x}^{2b}\mcH_{a-b}
\end{equation*}
\end{lemma}
\begin{proof}
The results follows from iteration of the fact that the action of the Laplace operator leads to $\Delta_x\brac{H_a(x)H_b(x)}=2\sum_{j=1}^m (\D_{x_j}H_a(x))(\D_{x_j}H_b(x))$.
\end{proof}
\begin{lemma}\label{Lemma_calc}
For all integers $k\geq 2$, we have:
\begin{equation*}
\Delta_x^{k-1}\inner{x,2\f_1}^2\inner{xux,2\f_1}^{k-2}=0
\end{equation*} 
\end{lemma}
\begin{proof}
Recalling that $\inner{xux,2\f_1}=\norm{x}^2\overline{w}_1-2\inner{u,x}\overline{z}_1$ with $z_1=x_1+ix_2$ and $w_1=u_1+iu_2$, we have for $k \geq 2$:
\begin{equation*}
\inner{x,2\f_1}^2\inner{xux,2\f_1}^{k-2}=\sum_{j=0}^{k-2} {k-2 \choose j} (-2)^{k-j-2}\norm{x}^{2j}\overline{z}_1^{k-j}\inner{u,x}^{k-j-2}\overline{w}_1^j.
\end{equation*}
The term $\inner{u,x}^{k-j-2}$ can be decomposed into harmonics in $x$. We have two different cases depending on the parity of $k-j$. For $k-j$ even, we have
\begin{equation*}
\overline{z}_1^{k-j}\inner{u,x}^{k-j-2}=\overline{z}_1^{k-j}\brac{H_{k-j-2}+ \ldots + \norm{x}^{k-j-2}H_0},
\end{equation*}
with $H_i \in \mcH_i$. Using the previous lemma, the term $\overline{z}_1^{k-j}H_{k-j-2}$ can be written as
\begin{equation*}
\overline{z}_1^{k-j}H_{k-j-2}=H'_{2k-2j-2}+\ldots+\norm{x}^{2k-2j-4}H'_{2}.
\end{equation*}
Similarly, we can write
\begin{equation*}
\norm{x}^2\overline{z}_1^{k-j}H_{k-j-4}=\norm{x}^2H^{\ast}_{2k-2j-4}+\ldots+\norm{x}^{2k-2j-6}H^{\ast}_{4},
\end{equation*}
which means that 
\begin{equation*}
\norm{x}^{2j}\overline{z}_1^{k-j}\inner{u,x}^{k-j-2}=\norm{x}^{2j}\widetilde{H}_{2k-2j-2}+ \ldots + \norm{x}^{2k-4}\widetilde{H}_2,
\end{equation*}
since $1 \leq j \leq k-2$, the maximal power of $\norm{x}^2$ that can appear is the one for $j=k-2$, i.e. 
\begin{equation*}
\norm{x}^{2k-4}\overline{z}_1^{k-j}\inner{u,x}^{k-j-2}=\norm{x}^{2k-4}\widetilde{H}_{2k-2j-2}+ \ldots + \norm{x}^{2k-4}\widetilde{H}_2,
\end{equation*}
which is clearly in the kernel of $\Delta_x^{k-1}$. The case where $k-j$ is odd can be computed in a similar way which completes the proof of the lemma. 
\end{proof}
\begin{lemma}
For all positive integers $k$, we have: 
\begin{equation*}
\Delta_x^k\inner{xux,2\f_1}^k=2^{2k-1} k! (2k+m-4)\frac{\Gamma\brac{k+\frac{m}{2}-2}}{\Gamma\brac{\frac{m}{2}-1}}\inner{u,2\f_1}^k.
\end{equation*}
\end{lemma}
\begin{proof}
In the proof of proposition \ref{Prop4}, it was shown that
\begin{align*}
\Delta_x^k\inner{xux,2\f_1}^k=& 2k(2k+m-4)\inner{u,2\f_1}\Delta_x^{k-1}\inner{xux,2\f_1}^{k-1} \\
&+4k(k-1)\norm{x}^{2-m-2k}\norm{u}^2\Delta_x^{k-1}\inner{x,2\f_1}^2\inner{xux,2\f_1}^{k-2} \\
=&2k(2k+m-4)\inner{u,2\f_1}\Delta_x^{k-1}\inner{xux,2\f_1}^{k-1},
\end{align*}
where the second equality follows from the previous lemma. Repeating these steps leads to 
\begin{equation*}
\Delta_x^k\inner{xux,2\f_1}^k=2^{k} k! (2k+m-4)(2k+m-6)\ldots m(m-2)\inner{u,2\f_1}^k,
\end{equation*} 
which can be simplified to obtain the desired result. 
\end{proof}
\begin{lemma}
For all positive integers $k$, we have: 
\begin{equation*}
\Delta_x^k\inner{u,x}\inner{x,2\f_1}\inner{xux,2\f_1}^{k-1}=2^{2k-1} k!\frac{\Gamma\brac{k+\frac{m}{2}-2}}{\Gamma\brac{\frac{m}{2}-1}}\inner{u,2\f_1}^k.
\end{equation*}
\end{lemma}
\begin{proof}
A straightforward computation shows that 
\begin{align*}
\Delta_x^k\inner{u,x}\inner{x,2\f_1}\inner{xux,2\f_1}^{k-1}=&2\inner{u,2\f_1}\Delta_x^{k-1}\inner{xux,2\f_1}^{k-1} \\
&+4(k-1)(k-2)\norm{u}^2\Delta_x^{k-1}\inner{u,x}\inner{x,2\f_1}^3\inner{xux,2\f_1}^{k-3} \\
&+4(k-1)\norm{u}^2\Delta_x^{k-1}\inner{x,2\f_1}^2\inner{xux,2\f_1}^{k-2} \\
&+2(k-1)(2k+m-6)\inner{u,2\f_1} \\
&\times \Delta_x^{k-1}\inner{u,x}\inner{x,2\f_1}\inner{xux,2\f_1}^{k-2}.
\end{align*}
The third term is zero, which was proven in lemma \ref{Lemma_calc}. Using the same notations as in the proof of that lemma, we have
\begin{equation*}
\inner{u,x}\inner{x,2\f_1}^3\inner{xux,2\f_1}^{k-3}=\sum_{j=0}^{k-3} {k-3 \choose j} (-2)^{k-j-3}\norm{x}^{2j}\overline{z}_1^{k-j}\inner{u,x}^{k-j-2}\overline{w}_1^j.
\end{equation*}
This means that also $\Delta_x^{k-1}\inner{u,x}\inner{x,2\f_1}^3\inner{xux,2\f_1}^{k-3}=0$. The remaining expression is given by
\begin{align*}
\Delta_x^k\inner{u,x}\inner{x,2\f_1}\inner{xux,2\f_1}^{k-1}=&2\inner{u,2\f_1}\Delta_x^{k-1}\inner{xux,2\f_1}^{k-1} \\
&+2(k-1)(2k+m-6)\inner{u,2\f_1} \\
&\times \Delta_x^{k-1}\inner{u,x}\inner{x,2\f_1}\inner{xux,2\f_1}^{k-2} \\
=&2^k (k-1)!(2k+m-6)(2k+m-8)\ldots m(m-2)\inner{u,2\f_1}^k \\
&+2(k-1)(2k+m-6)\inner{u,2\f_1} \\
&\times \Delta_x^{k-1}\inner{u,x}\inner{x,2\f_1}\inner{xux,2\f_1}^{k-2}.
\end{align*}
The proof then follows from induction. 
\end{proof}
\noindent
Putting everything together, we find 
\begin{equation*}
\begin{split}
\mcD_k\Phi(x,u)=&\brac{\brac{2-m+\frac{4k}{2k+m-2}}(2k+m-4)+\frac{8k}{2k+m-2}} \\
&\times \frac{\pi^{\frac{m}{2}}\Gamma\brac{k+\frac{m}{2}-2}}{\Gamma\brac{\frac{m}{2}-1}\Gamma\brac{\frac{m}{2}+k}}\delta\brac{x}\inner{u,2\f_1}^k,
\end{split}
\end{equation*} 
which can be simplified to 
\begin{equation}\label{label5}
\begin{split}
\mcD_k\Phi(x,u)=&\frac{4(4-m)\pi^{\frac{m}{2}}}{(2k+m-4)\Gamma\brac{\frac{m}{2}-1}}\delta\brac{x}\inner{u,2\f_1}^k,
\end{split}
\end{equation} 
We have thus reached the following conclusion, in which we use the notation $R(\omega)$ for the reflection $f(u) \mapsto R(\omega)f(u) = f(\omega\:u\:\omega)$: 
\begin{theorem}
The distribution
\begin{equation*}
e_k\brac{x}:=\frac{(2k+m-4)\Gamma\brac{\frac{m}{2}-1}}{4(4-m)\pi^{\frac{m}{2}}}\norm{x}^{2-m}R\brac{\omega} \in \mathcal{C}^{\infty}\brac{\mR_0^m,\mathrm{End}\brac{\mathcal{H}_k}}
\end{equation*}
satisfies, for every $H_k\brac{u} \in \mathcal{H}_k$, the following equation in distributional sense: 
\begin{equation*}
\mcD_ke_k\brac{x}H_k\brac{u}=\delta\brac{x}H_k\brac{u}.
\end{equation*}
\end{theorem}
\noindent
Let us then introduce the Fischer inner product on $\mathcal{H}_k$, by means of
\begin{equation*}
\comm{f,g}_F:=\left.\brac{\overline{f(\D_u)}g(u)}\right|_{u=0},
\end{equation*}
where we used the notation $\overline{\cdot}$ for complex conjugation. In order to obtain the fundamental solution for the higher spin Laplace operator, we will let the distribution $e_k\brac{x}$ act on the reproducing kernel $K_k(u,v)$ for $\mcH_k$ with respect to this inner product, satisfying the defining relation $\comm{K_k(u,v),H_k(u)}_F=H_k(v)$, for each $H_k(u) \in \mcH_k$. The reproducing kernel for $\mcH_k$ is given by a so-called Gegenbauer polynomial, i.e. $K_k(u,v):=\norm{u}^{2k}\norm{v}^{2k}C_k^{(\frac{m}{2}-1)}\brac{\frac{\inner{u,v}}{\norm{u}\norm{v}}}$, see e.g. \cite{AAR}, page 302. Hence, we have obtained our main result: 
\begin{theorem}
The fundamental solution for the higher spin Laplace operator $\mcD_k$ is defined as
\begin{equation*}
E_k\brac{x;u,v}:=e_k\brac{x}K_k(u,v) =\frac{(2k+m-4)\Gamma\brac{\frac{m}{2}-1}}{4(4-m)\pi^{\frac{m}{2}}}\norm{x}^{2-m}K_k(\omega \, u\, \omega,v)
\end{equation*}
\end{theorem}
\begin{remark}
Note that there are two special values for $k$: when $k=0$, we get the classical Laplace operator, i.e. $\mcD_0=\Delta_x$. The fundamental solution then becomes $E_k\brac{x}=c_0\norm{x}^{2-m}$, where the constant $c_0$ equals
\begin{equation*}
c_0=\frac{(m-4)\Gamma\brac{\frac{m}{2}-1}}{4(4-m)\pi^{\frac{m}{2}}}=-\frac{\Gamma\brac{\frac{m}{2}-1}}{4\pi^{\frac{m}{2}}}=\frac{\Gamma\brac{\frac{m}{2}}}{2(2-m)\pi^{\frac{m}{2}}}=\frac{1}{(2-m)A_m}.
\end{equation*}
For $k=1$, we get the generalised Maxwell operator, see e.g. \cite{ER}. The fundamental solution of $\mcD_1$ is given by $E_1\brac{x;u,v}=c_1\norm{x}^{2-m}K_1(\omega \, u\, \omega,v)$, where the constant $c_1$ is given by
\begin{equation*}
c_1=\frac{\brac{\frac{m}{2}-1}\Gamma\brac{\frac{m}{2}-1}}{2(4-m)\pi^{\frac{m}{2}}}=\frac{1}{(4-m)A_m},
\end{equation*}
which nicely corresponds to the one that was found in \cite{ER}.
\end{remark}

\section{Connection with the Rarita-Schwinger operator}
\label{RS}

Since the higher spin Laplace operator is the higher spin version of the Laplace operator one can wonder whether there is a connection with the Rarita-Schwinger operator, the simplest higher spin version of the Dirac operator. This question is inspired by the fact that 
\begin{equation}\label{refinement}
\ker_k\Delta_x \otimes \mS  \cong  \ker_{k}\D_x \oplus \ker_{k-1}\D_x\ ,  
\end{equation}
a relation which is known as the monogenic Fischer refinement (see \cite{BDS}): 
\begin{theorem}\label{F-refinement} 
If $H_k\brac{u} \in \mcH_k\brac{\mR^m,\mS}$ is a spinor-valued $k$-harmonic polynomial, we have that
\begin{equation*}
H_k\brac{u}=M_k\brac{u}+uM_{k-1}\brac{u},
\end{equation*}
with $M_j\brac{u}\in \mcM_j\brac{\mR^m,\mS}$. Both polynomials are uniquely determined by
\begin{align*}
M_{k-1}\brac{u}&=p_0H_k\brac{u}=-\frac{1}{2k+m-2}\D_uH_k\brac{u} \\
M_{k}\brac{u}&=p_1H_k\brac{u}=\brac{1+\frac{1}{2k+m-2}u\D_u}H_k\brac{u}.
\end{align*}
\end{theorem}
\noindent 
To investigate this, we recall the construction of the Rarita-Schwinger operator, see also \cite{BSSVL1, EVL}. Theorem \ref{F-refinement} allows us to construct the Rarita-Schwinger operator, which is an operator acting on functions taking values in the space of monogenic functions (this space again plays the role of the higher spin fields): 
\begin{theorem}
Let $f(x;u)$ be a function in $\mcC^{\infty}\brac{\mR^m,\mcM_k}$, then the Rarita-Schwinger operator 
\begin{equation*}
\mcR_k:\mcC^{\infty}\brac{\mR^m,\mcM_k} \longrightarrow  \mcC^{\infty}\brac{\mR^m,\mcM_k},
\end{equation*} 
is the unique (up to a multiplicative constant) conformally invariant first-order differential operator defined as
\begin{equation*}
\mcR_kf(x;u):=\brac{1+\frac{u\D_u}{m+2k-2}}\D_xf(x;u).
\end{equation*}
\end{theorem}
\noindent
The kernel of the Rarita-Schwinger operator $\mcR_k$ is given by (see \cite{BSSVL1} for a detailed proof):
\begin{equation*}
\mathrm{ker}_{\ell}\mcR_k\cong \bigoplus_{i=0}^{k} \bigoplus_{j=0}^{k-i} \mcS_{\ell-i+j,k-i-j}.
\end{equation*}
Now, recalling theorem \ref{ker_decomp} and using standard tensor product decomposition rules, one has that
\begin{equation}
\ker_{\ell}\mcD_k \otimes \mS \cong \ker_{\ell}\mcR_k \oplus \ker_{\ell-1}\mcR_k\oplus\ker_{\ell}\mcR_{k-1} \oplus\ker_{\ell-1}\mcR_{k-1}.
\label{decomposition4}
\end{equation}
At this point, the last statement is merely an isomorphism between irreducible representations for $\so(m)$, but we will prove that this is actually an equality. This then generalises expression (\ref{refinement}) to our higher spin case. First of all, it is clear that $p_1+up_0=\mathrm{Id}$, so we can write $\mcD_k=\mcD_k\brac{p_1+up_0}$ and we will find expressions for each of these terms. On $\mcM_k$-valued functions, the square of the Rarita-Schwinger operator reads:
\begin{align*}
\mcR_k^2 =& -\Delta_x+\frac{4}{(2k+m-2)^2}u\D_x\inner{\D_u,\D_x} -\frac{2}{2k+m-2}\set{u,\D_x}\inner{\D_u,\D_x} \\
&-\frac{4}{(2k+m-2)^2}\norm{u}^2\inner{\D_u,\D_x}^2 \\
=&-\mcD_k+\frac{4}{2k+m-2}\brac{\frac{1}{2k+m-4}-\frac{1}{2k+m-2}}\norm{u}^2\inner{\D_u,\D_x}^2 \\
&+\frac{4}{(2k+m-2)^2}u\D_x\inner{\D_u,\D_x},
\end{align*}
which means that 
\begin{equation*}
\mcD_k p_1=\brac{-\mcR_k^2+\frac{4}{(2k+m-2)^2}\brac{\frac{2\norm{u}^2}{2k+m-4}\inner{\D_u,\D_x}+u\D_x}\inner{\D_u,\D_x}}p_1.
\end{equation*}
Using the fact that 
\begin{equation*}
\inner{\D_u,\D_x}\mcR_k=\brac{\frac{2k+m-4}{2k+m-2}\D_x-\frac{2}{2k+m-2}u\inner{\D_u,\D_x}}\inner{\D_u,\D_x},
\end{equation*}
which was shown (and exploited) in \cite{BSSVL1}, we can also write this expression as:
\begin{equation}\label{pi1}
\mcD_k p_1=\brac{-\mcR_k+\frac{4}{(2k+m-2)(2k+m-4)}u\inner{\D_u,\D_x}}\mcR_kp_1.
\end{equation}
The projection $p_1$ ensures that each of the operators appearing at the right-hand side is well-defined: we have that $\mcR_k$ (respectively $u\inner{\D_u,\D_x}$) maps $\mcM_k$-valued functions to $\mcM_k$-valued functions (respectively $u\mcM_{k-1}$-valued functions). Also, from equation (\ref{pi1}), it is clear that we have for $f\in\mcC^{\infty}\brac{\mR^m,\mcM_k}$:
\begin{equation*}
\mcR_kf=0 \implies \mcR_{k-1}\inner{\D_u,\D_x}f\ =\ 0\ =\ \mcD_kf\ .
\end{equation*}
Next, we calculate 
\begin{align*}
\mcD_k up_0 &=\brac{\Delta_x-\frac{4}{2k+m-2}\brac{\inner{u,\D_x}-\frac{\norm{u}^2}{2k+m-4}\inner{\D_u,\D_x}}\inner{\D_u,\D_x}}up_0 \\
&=-\brac{u\mcR_{k-1}+\frac{4}{2k+m-2}\brac{\inner{u,\D_x}-\frac{\norm{u}^2}{2k+m-4}\inner{\D_u,\D_x}}}\mcR_{k-1}p_0,
\end{align*}
The projection $p_0$ again ensures that each of the operators appearing at the right-hand side is well-defined. 
The first term between brackets (respectively the second operator) maps $\mcM_{k-1}$-valued functions to $u\mcM_{k-1}$-valued functions (respectively $\mcM_k$-valued functions). Also, from this equation it is clear that for $f(x,u) =uf_0$ with $f_0 \in\ker\mcR_{k-1}$ one has that
\begin{equation*}
\mcD_k f = \mcD_k up_0(uf_0)=0.
\end{equation*}
The higher spin Laplace operator can thus be written as a combination of the following operators: 
\begin{align*}
-\mcR_k^2p_1&:\mcH_k\otimes \mS \longrightarrow \mcM_k \\
\frac{4}{(2k+m-2)(2k+m-4)}u\inner{\D_u,\D_x}\mcR_kp_1&:\mcH_k\otimes \mS \longrightarrow u\mcM_{k-1} \\
-u\mcR_{k-1}^2p_0&:\mcH_k\otimes \mS \longrightarrow u\mcM_{k-1} \\
-\frac{4}{2k+m-2}\brac{\inner{u,\D_x}-\frac{\norm{u}^2}{2k+m-4}\inner{\D_u,\D_x}}\mcR_{k-1}p_0&:\mcH_k\otimes \mS \longrightarrow \mcM_k \\
\end{align*}
Let us then construct the explicit embedding maps for decomposition (\ref{decomposition4}), turning the isomorphism into an equality. Recall that if $f \in \mcC^{\infty}\brac{\mR^m,\mcM_k} \cap \ker \mcR_k$, we have that $\mcD_kf=0$, which means that $\ker_{\ell}\mcR_k$ can be embedded with the identity map. Another straightforward embedding is the embedding of the space $\ker_{\ell}\mcR_{k-1}$, since $\mcD_1u f=0$ for functions $f(x)\in\ker_{\ell}\mcR_{k-1}$. In order to embed the space $\mathrm{ker}_{\ell-1}\mcR_k$, we note that one needs an embedding map which is homogeneous of degree one in $x$. In view of the fact that $\mcD_k f = \mcD_k p_1 f$, we easily find: 
\begin{equation*}
f \in \ker_{\ell-1}\mcR_k\ \Rightarrow\ \mcD_k x f = x\mcD_kf+2\brac{\D_x-\frac{2}{2k+m-2}u\inner{\D_u,\D_x}}f=0\ .
\end{equation*}
It is thus clear that a multiplication with $x$ does the job. For the final space $\ker_{\ell-1}\mcR_{k-1}$, we need an embedding map which is homogeneous of degree $(1,1)$ in $(x,u)$. To do so, we will make use of another conformally invariant operator, which is also known as a twistor operator and is defined as:
\begin{equation*}
p_1\pi_k\inner{u,\D_x}:\mcC^{\infty}\brac{\mR^m,\mcM_{k-1}}\longrightarrow \mcC^{\infty}\brac{\mR^m,\mcM_k}.
\end{equation*}
This twistor operator is explicitly given by, see e.g. \cite{Eelbode2} for more information:
\begin{equation*}
p_1\pi_k\inner{u,\D_x}:=\inner{u,\D_x}+\frac{1}{2k+m-2}\brac{u\D_x-\norm{u}^2\inner{\D_u,\D_x}}
\end{equation*}
The following result provides us with the desired combination for the embedding of $\ker_{\ell-1}\mcR_{k-1}$ into $\ker_{\ell}\mcD_{k}$: 
\begin{lemma}
The space $\ker_{\ell-1}\mcR_{k-1}$ can be embedded in the kernel of the higher spin Laplace operator (acting on spinor-valued functions) by means of: 
\begin{equation*}
\mcJ_R\Delta_x\mcJ_R\,p_1\pi_k\inner{u,\D_x}:\ker_{\ell-1}\mcR_{k-1} \longhookrightarrow \ker_{\ell}\mcD_k \otimes \mS.
\end{equation*}
Here $p_1\pi_k\inner{u,\D_x}$ is a twistor operator acting on functions taking values in the space of monogenics of degree $k-1$.
\end{lemma}
\begin{proof}
A tedious but straightforward calculation shows that $p_1\pi_k\inner{u,\D_x}$ maps $\ker_{\ell-1}\mcR_{k-1}$ to $\ker_{\ell-2}\mcR_k$. Since a solution of $\mcR_k$ is automatically a solution of $\mcD_k$, we only need an operator of homogeneity $(2,0)$ in $(x,u)$ to fix the homogeneity in $x$. It is clear from section \ref{PolSol} that $\mcJ_R\Delta_x\mcJ_R$ is the desired operator. 
\end{proof}
\noindent
Let us then restate the conclusion: 
\begin{theorem}
The space $\ker_\ell\mcD_k \otimes \mS$ of $\ell$-homogeneous polynomial null solutions for the higher spin Laplace operator $\mcD_k$ acting on the space $\mcC^\infty(\mR^m,\mcH_k \otimes \mS)$ decomposes as follows: 
\begin{align*}
\ker_{\ell}\mcD_k \otimes \mS &= \ker_{\ell}\mcR_k \oplus x\ker_{\ell-1}\mcR_k\oplus u\ker_{\ell}\mcR_{k-1} \\
&\oplus\big(\mcJ_R\Delta_x\mcJ_R\,p_1\pi_k\inner{u,\D_x}\big)\ker_{\ell-1}\mcR_{k-1}.
\end{align*}
\end{theorem}


\appendix
\section{Explicit proof of conformal invariance and ellipticity}
\label{ConfEll}
In this appendix, we first consider the conformal invariance in more detail. The following technical result was already mentioned in expression (\ref{special_transfo}): 
\begin{proposition}\label{propA1}
The following property holds for all $1 \leq j \leq m$:
\begin{equation*}
\mathcal{J}_R\D_{x_j}\mathcal{J}_R = 2\inner{u,x}\D_{u_j}-2u_j\inner{x,\D_u}+\norm{x}^2\D_{x_j}-x_j\brac{2\mE_x+m-2}
\end{equation*}
\end{proposition} 
\begin{proof}
Suppose $P_q(x,u)\in\mcP_{q}\brac{\mR^m,\mcH_k}$, i.e. $\mE_xP_q=qP_q$. The action of $\mathcal{J}_R$ on $P_q$ is given by:
\begin{equation*}
\mathcal{J}_R P_q=\norm{x}^{2-m}P_q\brac{\frac{x}{\norm{x}^2},\frac{xux}{\norm{x}^2}} :=\norm{x}^{2-m-2q}P_q\brac{x,v(x)}\ ,
\end{equation*}
where $v(x)$ is defined through the second equality. We thus get: 
\begin{equation*}
\begin{split}
\partial_{x_j}\mathcal{J}_R P_q=&\norm{x}^{2-m-2q}\brac{\dot{\partial_{x_j}}P_q\brac{\dot{x},v(x)}+\dot{\partial_{x_j}}P_q\brac{x,\dot{v}(x)}}\\ 
&+ (2-m-2q)x_j\norm{x}^{-m-2q}P_q.
\end{split}
\end{equation*}
Here, the dot on the argument of the function and partial derivative is to point out that the partial derivative only acts on that argument. The first and the last term are the same as for the $\mC$-valued case (using the inversion $\mcJ$), which leads to the generalised symmetry
\begin{equation*}
\mcJ \partial_{x_j} \mcJ = \norm{x}^2\partial_{x_j} - x_j(2\mE_x+ m - 2). 
\end{equation*}
The remaining term gives:
\begin{align*}
&\mathcal{J}_R\brac{\norm{x}^{2-m-2q} \dot{\partial_{x_j}}P_q\brac{x,\dot{v}(x)}} \\
&=\mathcal{J}_R\brac{\norm{x}^{2-m-2q} \sum_{i=1}^m\frac{\partial v_i}{\partial_{x_j}} \dot{\partial_{v_i}}P_q\brac{x,\dot{v}(x)}} \\
&=\mathcal{J}_R\brac{\norm{x}^{2-m-2q}\sum_{i=1}^m\frac{\partial}{\partial_{x_j}}\brac{u_i-2\frac{\inner{x,u}}{\norm{x}^2}x_i} \dot{\partial_{v_i}}P_q\brac{x,\dot{v}(x)}} \\
&=-2\mathcal{J}_R\brac{\norm{x}^{2-m-2q}\sum_{i=1}^m\left(\frac{u_j x_i}{\norm{x}^2} +\frac{\inner{x,u}\delta_{ij}}{\norm{x}^2}- 2\frac{\inner{x,u}x_ix_j}{\norm{x}^4}}\dot{\partial_{v_i}}P_q\brac{x,\dot{v}(x)}\right) \\
&=2\left(-\brac{u_j-2\frac{\inner{x,u}x_j}{\norm{x}^2}}\inner{x,\D_u}+\inner{x,u}\partial_{u_j}-2\frac{\inner{x,u}\inner{x,\D_u}}{\norm{x}^2}\right)
\end{align*}
Simplifying the last term completes the proof.
\end{proof}
\noindent
We then arrive at the main proposition, stating that the special conformal transformations are generalised symmetries of the higher spin Laplace operator:
\begin{proposition}
The special conformal transformations
\begin{equation*}
\mathcal{J}_R\partial_{x_j}\mathcal{J}_R := 2\inner{u,x}\partial_{u_j}-2u_j\inner{x,\D_u}+\norm{x}^2\partial_{x_j}-x_j\brac{2\mE_x+m-2},
 \end{equation*} 
 with $j\in\set{1,\cdots,m}$ are generalised symmetries of the higher spin Laplace operator. 
 \label{special conformal2}
\end{proposition}
\noindent
Using the fact that 
\begin{equation*}
\comm{AB,CD}=A\comm{B,C}D+AC\comm{B,D}+\comm{A,C}DB+C\comm{A,D}B,
\end{equation*}
we first can prove the following technical lemmas:
\begin{lemma}
For all $1\leq j\leq m$, we have:
\begin{equation*}
\comm{\Delta_x,\mathcal{J}_R\partial_{x_j}\mathcal{J}_R}=-4x_j\Delta_x+4\inner{u,\D_x}\partial_{u_j}-4u_j\inner{\D_u,\D_x}.
\end{equation*}
\end{lemma}
\begin{proof}
Follows from straightforward calculations. 
\end{proof}
\begin{lemma}
For all $1\leq j\leq m$, we have:
\begin{align*}
\comm{\inner{u,\D_x}\inner{\D_u,\D_x},\mathcal{J}_R\partial_{x_j}\mathcal{J}_R}&=2\norm{u}^2\partial_{u_j}\inner{\D_u,\D_x}-4x_j\inner{u,\D_x}\inner{\D_u,\D_x} \\ &+\brac{\inner{u,\D_x}\partial_{u_j}-u_j\inner{\D_u,\D_x}}\brac{2\mE_u+m-2}.
\end{align*}
\end{lemma}
\begin{proof}
Denoting $C_j := \mathcal{J}_R\partial_{x_j}\mathcal{J}_R$, we get: 
\begin{align*}
\comm{\inner{u,\D_x}\inner{\D_u,\D_x},C_j}=&2\inner{u,\D_x}\inner{x,\D_u}\partial_{x_j}+2\inner{u,x}\partial_{x_j}\inner{\D_u,\D_x}  \\
&-\inner{u,\D_x}\partial_{u_j}\brac{2\mE_x+m-2}-2\inner{u,\D_x}x_j\inner{\D_u,\D_x} \\
&+u_j\brac{2\mE_x+m-2}\inner{\D_u,\D_x}-2x_j\inner{u,\D_x}\inner{\D_u,\D_x}  \\
&+2\inner{u,\D_x}\brac{\mE_u+\mE_x+m}\partial_{u_j}+2\norm{u}^2\partial_{u_j}\inner{\D_u,\D_x}\\
&-2\inner{u,x}\partial_{x_j}\inner{\D_u,\D_x}-2\inner{u,\D_x}\partial_{x_j}\inner{x,\D_u} \\
&-2u_j\brac{\mE_u-\mE_x}\inner{\D_u,\D_x} \\
=&2\inner{u,\D_x}\inner{x,\D_u}\partial_{x_j}-\inner{u,\D_x}\partial_{u_j}\brac{m-2} \\
&-2\brac{x_j\inner{u,\D_x}-u_j}\inner{\D_u,\D_x}+u_j\brac{m-2}\inner{\D_u,\D_x} \\
&-2x_j\inner{u,\D_x}\inner{\D_u,\D_x}+2\inner{u,\D_x}\brac{\mE_u+m}\partial_{u_j} \\
&+2\norm{u}^2\partial_{u_j}\inner{\D_u,\D_x}-2u_j\inner{\D_u,\D_x}\brac{\mE_u-1} \\
&-2\inner{u,\D_x}\brac{\inner{x,\D_u}\partial_{x_j}+\partial_{u_j}}.
\end{align*}
Simplifying the last expression completes the proof.
\end{proof}
\noindent
\begin{lemma}
For all $1\leq j\leq m$, we have:
\begin{equation*}
\comm{\norm{u}^2\inner{\D_u,\D_x}^2,\mathcal{J}_R\partial_{x_j}\mathcal{J}_R}=-4x_j\norm{u}^2\inner{\D_u,\D_x}^2+2\norm{u}^2\partial_{u_j}\inner{\D_u,\D_x} \brac{2\mE_u+m-4}.
\end{equation*}
\end{lemma}
\begin{proof}
Again denoting $C_j := \mathcal{J}_R\partial_{x_j}\mathcal{J}_R$, we get: 
\begin{align*}
\comm{\norm{u}^2\inner{\D_u,\D_x}^2,C_j}&=4\norm{u}^2\inner{x,\D_u}\inner{\D_u,\D_x}\partial_{x_j}- 4\norm{u}^2x_j\inner{\D_u,\D_x}^2  \\
&-2\norm{u}^2\partial_{u_j}\inner{\D_u,\D_x}\brac{2\mE_x+m-2}-4u_j\inner{u,x}\inner{\D_u,\D_x}^2 \\
&+4\norm{u}^2\brac{\mE_x+\mE_u+m+1}\inner{\D_u,\D_x}\partial_{u_j}+4u_j\inner{u,x}\inner{\D_u,\D_x}^2 \\
&-4\norm{u}^2\inner{\D_x,\D_u}\partial_{x_j}\inner{x,\D_u} \\
&=4\norm{u}^2\inner{x,\D_u}\inner{\D_u,\D_x}\partial_{x_j}- 4\norm{u}^2x_j\inner{\D_u,\D_x}^2  \\
&-2\norm{u}^2\partial_{u_j}\inner{\D_u,\D_x}\brac{2\mE_x+m-2}-4u_j\inner{u,x}\inner{\D_u,\D_x}^2 \\
&+4\norm{u}^2\partial_{u_j}\inner{\D_u,\D_x}\brac{\mE_x+\mE_u+m+1-3}+4u_j\inner{u,x}\inner{\D_u,\D_x}^2 \\
&-4\norm{u}^2\inner{x,\D_u}\inner{\D_x,\D_u}\partial_{x_j}-4\norm{u}^2\partial{u_j}\inner{\D_x,\D_u} 
\end{align*}
There is no term with $\Delta_u$ involved because we work with $\mathcal{H}_k$-valued functions. Simplifying the last expression completes the proof.
\end{proof}
\noindent
Now we can put everything together, hereby again using the notation $C_j$: 
\begin{align*}
\mathcal{D}_kC_j&=C_j\mathcal{D}_k+\comm{\mathcal{D}_k,C_j} \\
&=\brac{C_j-4x_j}\mathcal{D}_k+4\brac{\inner{u,\D_x}\partial_{u_j}-u_j\inner{\D_u,\D_x}} \\
&-\frac{4}{2k+m-2}\brac{\inner{u,\D_x}\partial_{u_j}-u_j\inner{u,\D_x}}\brac{2\mE_u+m-2}-\frac{8}{2k+m-2m}\norm{u}^2 \partial_{u_j} \inner{u,\D_x} \\
&+\frac{8}{(2k+m-2)(2k+m-4)}\norm{u}^2\D_{u_j}\inner{\D_u,\D_x}\brac{2\mE_u+m-4} \\
&=\brac{C_j-4x_j}\mathcal{D}_k\ ,
\end{align*}
which completes the proof of proposition \ref{special conformal2}. Finally, we also prove that the higher spin Laplace operator is elliptic. We first start with the definition of ellipicity: 
\begin{definition}
A linear homogeneous differential operator of second order 
\begin{equation*}
\mcD:\mcC^{\infty}\brac{\mR^m,\mV_{\lambda}}\longrightarrow \mcC^{\infty}\brac{\mR^m,\mV_{\mu}},
\end{equation*}
where $\mV_{\lambda}$ and $\mV_{\mu}$ are vector spaces, is elliptic if for every non-zero vector $x\in \mR^m$ its principle symbol, which is a linear map $\sigma_x(\mcD):\mV_{\lambda}\longrightarrow \mV_{\mu}$ obtained by replacing its partial derivatives $\D_{x_j}$ with the corresponding variables $x_j$, is a linear isomorphism. 
\end{definition}
\begin{theorem}\label{Elliptic}
The higher spin Laplace operator, which is explicitely given by 
\begin{equation*}
\mcD_k:=\Delta_x-\frac{4}{2k+m-2}\brac{\inner{u,\D_x}-\frac{\norm{u}^2}{2k+m-4}\inner{\D_u,\D_x}}\inner{\D_u,\D_x}
\end{equation*} 
is an elliptic operator if $m>4$.
\end{theorem}
\begin{proof}
To prove the theorem, we will show that for fixed $x\in \mR_0^m$ the symbol of the higher spin Laplace operator, which is given by
\begin{equation*}
\sigma_x(\mcD_k)=\norm{x}^2-\frac{4}{2k+m-2}\brac{\inner{u,x}-\frac{\norm{u}^2}{2k+m-4}\inner{x,\D_u}}\inner{x,\D_u}: \mcH_k\longrightarrow\mcH_k,
\end{equation*}
is a linear isomorphism. As the symbol is clearly a linear map, it remains to be proven that the map is injective. To do so, we will need a clever choice for a basis for $\mcH_k\brac{\mR^m,\mC}$, which will be obtained using the classical CK extension for harmonic polynomials \cite{LSVL}. Therefore, we need the classical Kelvin inversion $\mcJ$, which is given by 
\begin{equation*}
\mcJ:\;\mcC^{\infty}\brac{\mR^m,\mC}\longrightarrow\mcC^{\infty}\brac{\mR_0^m,\mC}: f(u)\mapsto \mcJ\comm{f}(u):=\norm{u}^{2-m}f\brac{\frac{u}{\norm{u}^2}}
\end{equation*} 
In this case, we also have that 
\begin{equation*}
\mathfrak{sl}(2)\cong\Span\brac{\mathcal{J}\partial_{u_j}\mathcal{J},\partial_{u_j},2\mE_u+m-2},
\end{equation*}
where $\mathcal{J}\partial_{u_j}\mathcal{J}$ is explicitly given by
\begin{equation*}
\mathcal{J}\partial_{u_j}\mathcal{J}=\norm{u}^2\D_{u_j}-u_j\brac{2\mE_u+m-2}
\end{equation*}
For $x\in \mR_0^m$ fixed, we have 
\begin{equation*}
\mathcal{J}\inner{x,\D_u}\mathcal{J}=\norm{u}^2\inner{x,\D_{u_j}}-\inner{u,x}\brac{2\mE_u+m-2},
\end{equation*}
which means that we can write the symbol of $\mcD_k$ as
\begin{equation*}
\sigma_x(\mcD_k)=\norm{x}^2\brac{1+\frac{4}{(2k+m-2)(2k+m-4)}\brac{\mathcal{J}\inner{\omega,\D_u}\mathcal{J}}\inner{\omega,\D_u}}
\end{equation*}
The branching rules for the Lie algebra $\so(m)$ state that when we restrict the action on the irreducible representation with highest weight $\lambda=(k,0\ldots,0)$ to $\so(m-1)$, we get the following decomposition:
\begin{equation*}
(k,0,\ldots,0) \bigg|_{\so(m-1)}^{\so(m)}=\bigoplus_{j=0}^k(k-j,0,\ldots,0).
\end{equation*}
This means that an arbitrary harmonic polynomial $H_k(u)\in \mcH_k\brac{\mR^m,\mC}$ can be written as
\begin{equation*}
H_k(u)=\sum_{j=0}^k\brac{\mathcal{J}\inner{\omega,\D_u}\mathcal{J}}^jH^{\ast}_{k-j}(u),
\end{equation*}
where $H^{\ast}_{k-j}(u)\in\mcH_{k-j}\brac{\mR^{m},\mC}$ such that $\inner{\omega,\D_u}H^{\ast}_{k-j}(u)=0$. The right-hand side of the equation is then clearly invariant under the action of $\so(m-1)$, where $\so(m-1)$ has to be understood as the Lie algebra corresponding to the subgroup of SO$(m)$ containing rotations in the hyperplane perpendicular to $\omega\in \mR^m$. Using the relation 
\begin{equation*}
\comm{\inner{\omega,\D_u},\brac{\mathcal{J}\inner{\omega,\D_u}\mathcal{J}}^j}=-2\brac{\mathcal{J}\inner{\omega,\D_u}\mathcal{J}}^{j-1}\brac{2\mE_u+m-3},
\end{equation*}
which is a relation in the universal enveloping algebra $\mcU(\mathfrak{sl}(2))$ that can be proved by induction, the equation $\sigma_x(\mcD_k)H_k(u)=0$ leads to the following system:
\begin{equation*}
H_k(u)+\sum_{j=1}^k\brac{1-\frac{4j(2k+m-j-3)}{(2k+m-2)(2k+m-4)}}\brac{\mathcal{J}\inner{\omega,\D_u}\mathcal{J}}^jH^{\ast}_{k-j}(u)=0.
\end{equation*}
Since the polynomials $H^{\ast}_{k-j}(u)\in\mcH_k\brac{\mR^{m},\mC}$ are linearly independent for $1\leq j \leq k$, we have that either $H^{\ast}_{k-j}(u)=0$ for all $1\leq j \leq k$, which means that $\ker{\sigma_x(\mcD_k)}=0$ or that 
\begin{align*}
1-\frac{4j(2k+m-j-3)}{(2k+m-2)(2k+m-4)}&=0, \\
\iff (2k+m-4)(2k+m-2)-4j(2k+m-j-3)&=0.
\end{align*}
The latter is a polynomial of second order in $m$ and has two roots: 
\begin{equation*}
m=-2(k-j-1) \quad \text{and} \quad m=-2(k-j-2).
\end{equation*}
It is easy to see that for $k \in \mN$ fixed, only $m\leq 4$ causes trouble. This means that $\ker{\sigma_x(\mcD_k)}=0$ whenever $m>4$, which proves the statement.
\end{proof}
\noindent


\end{document}